\documentclass[]{article}

\usepackage{amssymb}	
\usepackage{amsthm}
\usepackage{color}
\usepackage{yhmath}
\usepackage{array}

\usepackage{amsfonts} 
\newcommand{\GZ}{{\mathbb{Z}}}
\newcommand{\GL}{{\mathfrak{L}}}
\newcommand{\GB}{{\mathfrak{B}}}
\newcommand{\GR}{{\mathbb{R}}}

\newtheorem{theorem}{Theorem}[section]
\newtheorem{corollary}{Corollary}[section]
\newtheorem{proposition}{Proposition}[section]
\theoremstyle{definition}
\newtheorem{definition}[theorem]{Definition}
\newtheorem{remark}{Remark}
\newtheorem{example}{Example}

\definecolor{darkviolet}{rgb}{0.58, 0.0, 0.83}

\usepackage{hyperref}
\usepackage[capitalise]{cleveref}

\date{\today}

\title{On Polynomial Modular Number Systems over $\mathbb{Z}/\lowercase{p}\mathbb{Z}$}

\author{Jean Claude Bajard \thanks{Sorbonne Universit\'e, CNRS, Inria, Institut de Math\'ematiques de Jussieu -- Paris Rive Gauche, France
({jean-claude.bajard@sorbonne-universite.fr})}
\and
 J\'er\'emy~Marrez \thanks{Sorbonne Universit\'e, CNRS, Laboratoire d'informatique de Paris 6, Paris, France
({jeremy.marrez@sorbonne-universite.fr})}
\and
Thomas~Plantard \thanks{University of Wollongong, Institute of Cybersecurity and Cryptology, Wollongong, Australia
({thomaspl@uow.edu.au})}
 \and
Pascal~V\'eron \thanks{Universit\'e  de Toulon, Institut de Math\'ematiques,  Toulon, France
({pascal.veron@univ-tln.fr})}
}

\begin{document}

	\maketitle

\begin{abstract}
Since their introduction in 2004,  Polynomial Modular Number Systems (PMNS) have become a very interesting tool for implementing cryptosystems relying on modular arithmetic in a secure and efficient way. However, while their implementation is simple, their parameterization is not trivial and relies on a suitable choice of the polynomial on which the PMNS  operates. The initial proposals were based on particular binomials and trinomials. But these polynomials do not always provide systems with interesting characteristics  such as small digits, fast reduction, etc.

 In this work, we study a larger family of polynomials that can be exploited to design a safe and efficient PMNS.  To do so, we first state  a complete existence theorem for PMNS which provides bounds on the size of the digits for a generic polynomial, significantly improving previous bounds.  
Then, we present classes of suitable polynomials which provide numerous PMNS for safe and efficient arithmetic. 
\end{abstract}

\section{Introduction}
\label{sec:introduction}

\subsection*{Context of the modular arithmetic} Modular arithmetic is at the core of modern cryptography~\cite{SP2018}. Modular operations (essentially multiplication and addition) appear in most of today's public key cryptography. Widely used cryptographic protocols such as RSA~\cite{RSA1978}, DSA~\cite{NIST2009} and their counterparts based on elliptic curves~\cite{Mil1985,Kob1987} are at the core of modern communication. The main cost of all these cryptosystems is due to modular arithmetic.
Their potential successors, currently competing in the post-quantum cryptography standardization contest organized by the U.S. National Institute of Standards and Technology NIST~\cite{AAACDKLMMPPRS20}, also rely heavily on modular arithmetic.
As an example, lattice based proposals such as Kiber~\cite{BDKLLSSSS18}, NTRU~\cite{HPS1998}, Saber~\cite{AKSV2018}, Falcon~\cite{PFHKLPRSWZ17}, or isogeny based key exchange (SIKE~\cite{SIKE2019}) rely all on fast modular arithmetic. Furthermore, pairing based cryptography offers revolutionary protocols~\cite{BF2001} which rely as well on modular arithmetic on large moduli.

\subsection*{Specific modular arithmetic} As improving modular arithmetic has such a wide impact on the efficiency of modern cryptographic protocols,  special classes of moduli have been investigated. These special moduli are generally inspired by Mersenne numbers (integers of the form $2^m - 1$) to perform a modular reduction as fast as possible, namely Pseudo Mersenne~\cite{Cra1992}, Generalized Mersenne~\cite{Sol1999}. Other ones have been created to be particularly efficient when used with some specific algorithm. As an example, Montgomery-friendly primes~\cite{Ham2012,BCHL2013,BD2021} have been created to be operated with Montgomery reduction~\cite{Mon1985}. 
However, these classes are by definition limited and multiple cryptosystems require a free choice of  the moduli on which they operate.

\subsection*{The origin of PMNS}
To obtain efficient modular arithmetic for all moduli, and not only for a class of special moduli, the PMNS\cite{Pla2005} were proposed as an effective representation system.
%When operating modulo some integer $p$, PMNS offers to change the number representation used to one dedicated to $p$: if $p$ is not part of an efficient class of moduli, PMNS allows to operate on a number system where $p$ becomes an efficient modulo.
They operate without carry propagation and offer both the advantages of fast polynomial arithmetic and easy parallelization for  arbitrary moduli $p$.
Specifically, a PMNS is a modular system, where any integer $a$ modulo $p$ (which is not necessarily a prime) is represented as a polynomial $A(X)$ of degree smaller than a fixed integer $n$. Modular multiplication and addition of two integers $a$ and $b$ in $\mathbb{Z}/p\mathbb{Z}$ are then computed using their representatives $A(X)$ and $B(X)$ in the PMNS. The coefficients of the polynomials are the digits and are bounded by an integer $\rho$, which is small relatively to $p$ $(\rho \simeq p^{1/ n})$. The construction of such systems is based on sparse polynomials whose roots $\gamma$ are used as  radices for this kind of positional representation, that is to say,  $A(\gamma) \equiv a \pmod{p}$. 
The interest of these sparse polynomials lies in the efficiency of the spawned modular arithmetic. 
The operations in PMNS are  done in two steps. First, the operations are carried out on polynomials modulo a sparse polynomial $E(X)$, called reduction polynomial, which is of degree $n$, and this reduction ensures that the degree of the result is smaller than $n$. 
In other words, to compute $a \odot b$  ($\odot$ representing an addition or a multiplication), one computes $C(X) = A(X) \odot B(X) \bmod E(X)$.
Then, a coefficient reduction is performed involving a lattice associated with the system \cite{HPS2011, PSZ2015,Gal2012}; this operation guarantees that the coefficients of the result $C(X)$ are bounded by $\rho$. 

%%%%%
A method for constructing a prime $p$ which has an efficient PMNS, has been published in 2004  \cite{BIP2015s}. The system is built from two sparse polynomials with good reduction properties (one is the reduction polynomial $E(X)$, the other one is used for the coefficient reduction), in order to derive the corresponding integer $p$ through the computation of a resultant, and also of one root $\gamma$. 
In order to be able to work with an arbitrary $p$, prime or not, another approach has been developed in  \cite{BIP2015a}  by constructing PMNS from an integer $p$, a number of digits $n$ and an integer polynomial $E(X)$ of the form $E(X) = X^n + aX + b$ satisfying some assumptions. Moreover, this result guarantees the existence of a PMNS  with a  bound on the digit size  $\rho$ allowing the representation of all numbers modulo $p$. Nevertheless, building such systems for a given $p$ is not trivial.

The structure of the reduction polynomial $E(X)$ gives the complexity of the polynomial reduction. Then, with   $p$ and a root $\gamma$  of $E(X)$ modulo $p$, we can define an associated lattice which allows to define the bound  $\rho$  on the coefficients of the representation and also provides the method of reduction of the coefficients. Therefore, it is interesting, for a given $p$, to find polynomials $E(X)$  giving efficient polynomial modular reductions and  roots $\gamma$ to define the associated lattice and the reduction of the coefficients.
%\color{red}
\subsection*{PMNS in a cryptographic context}
The efficiency of this system of representation was the subject of an in-depth study in \cite{DDV2020} for binomials $E(X) = X^n -\lambda$. Such a representation system is called an AMNS (Adapted Modular Number System) \cite{BIP2015s}.
 It has been observed that, for primes $p$ whose size fits the standard sizes used in elliptic curve cryptography (ECC), the AMNS representation allows to compute modular multiplications in a much more efficient way than the classical libraries OpenSSL and GnuMP (even if using for the latter the low level arithmetic functions and the undocumented Montgomery multiplication function). Later, this study has been confirmed in \cite{MPHELL} which described a specific library for ECC, named MPHELL, and compared it with other dedicated cryptographic libraries. The results show that on a 64-bit architecture, the AMNS representation gives the best results inside MPHELL for ECDSA/EdDSA signatures (generation and verification). Moreover, it offers also competitive timings on an ARM v8 architecture or a STM32F4 board. In \cite{sike}, the authors extend the AMNS representation system to $\mathbb{F}_{p^k}$ and show how it can be used in order to improve the performances of SIKE \cite{SIKE2019}, one of the alternate KEM candidate of the NIST post-quantum standardization process \cite{nist}. A first hardware implementation of the AMNS is described in \cite{Asma2021}. To end, it is shown in \cite{DDEMMV2019,Negre_2021} that some ``random steps'' can be injected in AMNS multiplication in order to resist to a side channel analysis.
\color{black}
\subsection*{Motivation and main results} The major motivation and result of this paper is an effective construction of  efficient PMNS for any integer $p$. The efficiency is measured in particular by the minimality of the digit size $\rho$ 
which depends on a reduced basis of the associated lattice that we explicitly construct. In \cref{sec:theorem}, we give  bounds and properties  and used them in \cref{sec:classes,sec:number} to define what is a suitable polynomial for PMNS.
The main results can be summarised as follows~:
\begin{enumerate}
\item   \Cref{RhoBoundTheorem} lays down critical result on PMNS existence. It relates the digit size $\rho$ to the infinity norm of the transpose of a reduced basis (seen as a matrix) of the associated  lattice. 
The reduction criterion consists, in this context of  PMNS, in searching for a basis such that the infinity norm of its transpose is close to a minimal.

\item    In  \cref{PropInversible}, we first construct  a reduced basis for a sublattice built from a short vector of the initial associated lattice.   \Cref{Prop:irreducbound} specifies this point when $E(X)$ is an irreducible polynomial. In this case,  we give a bound for the digit size $\rho$ depending only on $p$ and $E(X)$. Then, \cref{CorIrreduc,CorIrreducb} provide concrete  construction methods for reduced lattice bases.

 \item Then we introduce effective constructions of efficient PMNS   introduced in    \cref{sec:classes,sec:number}. We provide multiple classes of polynomial $E(X)$ over which PMNS can be efficiently used, with studies on both their irreducibility and the size of the set of their roots in $\mathbb{Z} / p\mathbb{Z}$, two key parameters for their usability.

\end{enumerate}

% Another important result is the effective constructions of efficient PMNS   introduced in    \cref{sec:classes,sec:number}. We provide multiple classes of polynomial $E(X)$ over which PMNS can be efficiently used, with studies on both their irreducibility and the size of the set of their roots in $\mathbb{Z} / p\mathbb{Z}$, two key parameters for their usability.

\subsection*{Organization of the paper}
This paper is organized as follows: \cref{sec:lattice,sec:pmns} recall the necessary background respectively on lattice theory and PMNS. Then \cref{sec:theorem} presents   theorems, propositions and their corollaries, which provide criteria for constructing concrete efficient PMNS for any $p$.
In \cref{sec:classes}, we specify what is a suitable reduction polynomial $E(X)$, and propose  main classes of suitable irreducible  polynomials; they allow efficient reductions, and their roots can be clearly identified in a finite prime field $\mathbb{Z}/p \mathbb{Z}$.
    \Cref{sec:number} studies the number of roots in a finite prime field $\mathbb{Z}/p \mathbb{Z}$ of the  reduction polynomial $E(X)$.

\section{Lattice Basics} 
\label{sec:lattice}
Lattice theory, also known as  geometry of numbers, was introduced by H. Minkowski in 1896~\cite{Min1896}.

A comprehensive discussion on the basics of lattice theory is presented in~\cite{Cas1959,Lov1986,CS1988}. We present in this section only the different definitions and results useful for the comprehension of our paper.
%but the mathematical tool is briefly explained below. %Cas1959,Lov1986,   
                                             
\begin{definition}[Lattice]
A \textit{lattice} $\GL$ is a discrete subgroup of ${\GR}^n$,  that is, the set of all the integral combinations of $d \leqslant n$ linearly independent vectors over $\GR$:
\[ 
\GL = \GZ \,b_1 + \dots + \GZ\,b_{d}, \quad b_i \in \GR^n.
 \]
Here, $B=(b_1,...,b_d)$ is called a basis of $\GL$ and $d$, the \textit{dimension} of $\GL$. 
We note  $\GL(B)$   a lattice of basis $B$.
If $d=n$, the lattice is called \textit{full-rank}.
\end{definition}

The \textit{determinant} of $\GL$ defined by $\det{\GL}=\sqrt{\det{\left ( BB^T \right )}}$ is invariant for any basis $B$ of $\GL$.

Lattice theory problems are based on minimising the distance between vectors. The natural norm used in lattice theory is the euclidean norm.  The \textit{euclidean norm} of  a vector $v$ is computed by 
\[
 \|v\|_2= \sqrt{ \sum_{i=1}^nv_i^2}.
 \]
 Other $l_p-$norm, $\|v\|_p= \left ( \sum_{i=1}^n|v_i|^p \right )^{1/p},$  can also be used. If $p=\infty$, then the norm is called the max-norm $\|v\|_\infty= max_{i=1}^n|v_i|.$

One of the most studied lattice problems is the Shortest Vector Problem (SVP).

\begin{definition}[SVP]
  Given a lattice $\GL$,   solving the Shortest Vector Problem, amounts to finding a vector $u \in \GL$ such that $\forall v \in \GL\setminus\{0\}, 0<\|u\| \leqslant \|v\|$ for a given norm $\|.\|$.

  The norm $\|u\|$ of such vector $u$ is called the \textsl{first minimum} and is denoted as $\lambda_1$.
Moreover, $\lambda_i$ will represent the norm of the $i^{th}$ minimum (the minimum norm of $i$ linearly independent vectors). \\
If  the norm is not specified, one will assume $\lambda_i$ to be the $i^{th}$ minimum of $\GL$ for the Euclidean norm, $\lambda_{i,p}$ will represent the $i^{th}$ minimum of $\GL$ for the $l_p$ norm.
\end{definition}

In 1998, M. Ajtai~\cite{Ajtai98} proved that SVP is NP-hard under a randomized reduction. Therefore, the best algorithm to compute SVP in polynomial space uses exponential time. It was proposed by R. Kannan in 1983 and relies on strongly reducing  each vector of a basis by recursion. It is often referenced as HKZ for Hermite-Korkin-Zolotareff. Its current best time estimation is at $2^{O(d)}d^\frac{d}{2e}$~\cite{HS2007}. Furthermore, some polynomial solutions exist as well such as LLL~\cite{LLL1982} or BKZ~\cite{Sch1994}. 
However, these solutions return vectors whose norm is equal to  the size of the shortest vector   times  an exponential factor.

Nevertheless, certain bounds do exist on the first minimum and were given by Minkowski's initial work:
\begin{theorem}[Minkowski]
  Let $\GL$ a lattice of dimension $d$, then
 \[
 \lambda_{1,\infty} \leqslant (\det{\GL})^{\frac{1}{d}}.
 \]
\end{theorem}

This bound is tight in max-norm. However, it is still an open problem for the Euclidean Norm.
A second key problem of lattice theory is the Closest Vector Problem (CVP).

\begin{definition}[CVP]
  Given a lattice $\GL$ and a vector $w$, to solve CVP is to find a vector $u \in \GL$ such that $\forall v \in \GL, \|w-u\| \leqslant \|w-v\|$. The quantity $\|w-u\|$ is noted $dist(w,\GL)$.
\end{definition}

The problem CVP is NP-Hard as well~\cite{VEB1981}.
Finally, a key invariant has been studied to try to evaluate the orthogonality of a lattice, i.e., the \textit{Covering Radius}.

\begin{definition}[Covering Radius]
  Let $\GL$ be a full rank lattice. The covering radius of $\GL$, noted $\mu(\GL)$, is the supremum of distances between any vector of $\GR^d$ and $\GL$, i.e.,
  \[
  \mu(\GL)= \max_{v \in \GR^d} dist(v,\GL)\,.
  \]
\end{definition}

No polynomial algorithm exists to find the covering radius~\cite{GMR2004}. However, we know that for any $l_p$-norm~\cite{GMR2004}, we have 
\[
\mu_p(\GL) \geqslant \frac{\lambda_{d,p}}{2}.
\]
For the Euclidean norm, we know that $\mu_2(\GL) \leqslant \sqrt{d} \lambda_d(\GL)$~\cite{GMR2004}. By simple norm relation, we obtain that $\mu_\infty(\GL) \leqslant \sqrt{d} \lambda_{d}(\GL).$

\section{Polynomial Modular Number System} 
\label{sec:pmns}

In this section, we recall  basic definitions and results on PMNS.

  %\paragraph{The origin}
          
%  The Polynomial Modular Number Systems came from the approach of Solinas  \cite{sol1999} for modular reduction modulo $p$ where $p$ is a generalized pseudo-Mersenne number. This approach take into account a relation between $p$ and the radix $\beta$ of a \emph{classical positional number system} (i.e., $p=f(\beta)$, with $f(t)=t^d-(c_1t^{d-1}+dots+c_d)$   ). 
%Unfortunately, it is not obvious to find  generalized pseudo-Mersenne prime numbers. To get more opportunity, a new kind of representation was introduced in \cite{BIP2015a}, where the base, for a given $p$, deeply depends of the choice of the reduction polynomial $E(X)$.

\begin{definition}[Polynomial Modular Number System]
\label{def:pmns}
Let $p\geqslant 3$, $n\geqslant 2$, $\gamma\in[1,p-1]$ and $\rho \in[1,p-1]$ be  integers. Let $E(X)\in \mathbb{Z}[X]$ be a monic polynomial of degree $n$ that satisfies $E(\gamma) \equiv 0   \pmod{p}  $.
A {Polynomial Modular Number System} (PMNS) is a set $\GB\subset{\mathbb Z}[X]$ such that:
\begin{enumerate}
    \item $\forall A(X)\in\GB$, deg$(A(X)) < n$,
    \item $\displaystyle \forall A(X)=\sum_{i=0}^{n-1}a_iX^i\in\GB$, $ -\rho < a_i < \rho$ for all $i$,
    \item  $\forall a\in\{0,\dots p-1\}$,  $\exists A(X)\in\GB$ such that 
$A(\gamma)\equiv a\pmod p$.
\end{enumerate}
%$(x_0, \dots,x_{n-1})$ with $\displaystyle x \equiv \sum_{i=0}^{n-1} x_i \gamma^i \pmod{p},$ where 
The polynomial $E(X)$ is called \emph{reduction polynomial} with respect to $p$.  
\end{definition}

A PMNS is thus a system of representation for elements in  $\mathbb{Z}/p\mathbb{Z}$ where \\
$\displaystyle a\in\mathbb{Z}/p\mathbb{Z} \mbox{ with } a \equiv \sum_{i=0}^{n-1}a_i\gamma^i  \pmod p = A(\gamma)\bmod p \mbox{ and } -\rho < a_i < \rho \mbox{ for all }.$

It looks a priori like the classic $\gamma$-ary positional system but since the  $\gamma^i \bmod{p}$  are not ordered, there is no obvious way to compare two representatives $A(X)$ and $B(X)$ without computing $A(\gamma)\bmod p$ and $B(\gamma)\bmod p$. This is clearly shown in \cref{ex:pmns}.

 Throughout this paper, we use the  notation  $\GB=(p,n,\gamma,\rho)_E$ to recall that the PMNS $\GB$ is determined by these five parameters. Also, with a polynomial  $A(X) = a_0+a_1 X+ \dots+a_{n-1}X^{n-1}$ we associate the vector $A=(a_0,a_1,\dots,a_{n-1})$. We will switch  between both notation when it is best suited for comprehension.

Operations in $\GB$ are first done modulo $E(X)$,  and then a coefficient reduction process is performed, by subtracting an appropriate polynomial having $\gamma$ as root modulo $p$, to guarantee that all the coefficients are bounded by $\rho$ in absolute value \cite{PSZ2015,DDV2020}.

\begin{example}
  \label{ex:pmns}  \Cref{ex:pmns1} shows how to represent elements of $\mathbb{Z}/31\mathbb{Z}$ as polynomials of degree lower or equal to  $3$ and  coefficients belonging to   $\{ -1, 0, 1\}$.% (i.e., $\rho = 2$ ). The reduction polynomial is $E(X)=X^4-2$ and  ${\gamma = 15}$ is a root of $E(X)$. 
  
%$p = 31$, $n = 4$, ${\gamma = 15}$, $\rho = 2$, and $E(X)=X^4-2$, for representing the elements of $\mathbb{Z}/31\mathbb{Z}$ as polynomials with $4$ coefficients belonging to \{ -1, 0, 1\} (see Table \cref{ex:pmns1}).    We note that $\gamma^{4}-2\equiv0 \bmod{31}$. 

\begin{table}[!h]
  \centering {\scriptsize
\center
%\begin{tabular}{|*{6}{>{\centering\arraybackslash} p{1.75cm}|}}
\begin{tabular}{|@{}c@{}|@{}c@{}|@{}c@{}|@{}c@{}|@{}c@{}|@{}c@{}|}
  \hline
  0 	 & 1	       & 2	     & 3 	  & 4 		 & 5 \\
  \hline
        (0, 0, 0, 0)
   &
        (1, 0, 0, 0)
   &
        (-1, 1, -1, 1)
   &
   \begin{tabular}[t]{c}       
        (-1, -1, -1, 1)\\
        (-1, 0, 0, -1)\\
        (-1, 0, 1, 1)\\
        (0, 1, -1, 1)
         \end{tabular} 
   &
   \begin{tabular}[t]{c}        (0, -1, -1, 1)\\
        (0, 0, 0, -1)\\
        (0, 0, 1, 1)\\
        (1, 1, -1, 1)  \end{tabular}
   &
   \begin{tabular}[t]{c}        (1, -1, -1, 1)\\
        (1, 0, 0, -1)\\
        (1, 0, 1, 1)  \end{tabular}
   \\
   \hline
    6 	 & 7	       & 8	     & 9 	  & 10 		 & 11 \\
  \hline

        (-1, 1, -1, 0)
   &
   \begin{tabular}[t]{c}        (-1, -1, -1, 0)\\
        (-1, 0, 1, 0)\\
        (0, 1, -1, 0)  \end{tabular}
   &
   \begin{tabular}[t]{c}        (0, -1, -1, 0)\\
        (0, 0, 1, 0)\\
        (1, 1, -1, 0)  \end{tabular}
   &
    \begin{tabular}[t]{c}       (1, -1, -1, 0)\\
        (1, 0, 1, 0)  \end{tabular}
   &
     \begin{tabular}[t]{c}      (-1, 1, -1, -1)\\
        (-1, 1, 0, 1)  \end{tabular}
   &
  \begin{tabular}[t]{c}         (-1, -1, -1, -1)\\
        (-1, -1, 0, 1)\\
        (-1, 0, 1, -1)\\
        (0, 1, -1, -1)\\
        (0, 1, 0, 1)  \end{tabular}
   \\
   \hline
    12 	 &13       & 14	     & 15 	  & 16	 &17 \\
  \hline
  \begin{tabular}[t]{c}       (0, -1, -1, -1)\\
        (0, -1, 0, 1)\\
        (0, 0, 1, -1)\\
        (1, 1, -1, -1)\\
        (1, 1, 0, 1)  \end{tabular}
   &
   \begin{tabular}[t]{c}        (1, -1, -1, -1)\\
        (1, -1, 0, 1)\\
        (1, 0, 1, -1)  \end{tabular}
   &
        (-1, 1, 0, 0)
   &
    \begin{tabular}[t]{c}       (-1, -1, 0, 0)\\
        (0, 1, 0, 0)  \end{tabular}
   &
  \begin{tabular}[t]{c}         (0, -1, 0, 0)\\
        (1, 1, 0, 0)  \end{tabular}
   &
        (1, -1, 0, 0)
   \\
   \hline
   18 	 & 19	       &20	     & 21 	  & 22	 & 23 \\
  \hline

    \begin{tabular}[t]{c}       (-1, 0, -1, 1)\\
        (-1, 1, 0, -1)\\
        (-1, 1, 1, 1)  \end{tabular}
   &
    \begin{tabular}[t]{c}       (-1, -1, 0, -1)\\
        (-1, -1, 1, 1)\\
        (0, 0, -1, 1)\\
        (0, 1, 0, -1)\\
        (0, 1, 1, 1)  \end{tabular}
   &
    \begin{tabular}[t]{c}       (0, -1, 0, -1)\\
        (0, -1, 1, 1)\\
        (1, 0, -1, 1)\\
        (1, 1, 0, -1)\\
        (1, 1, 1, 1)  \end{tabular}
   &
    \begin{tabular}[t]{c}       (1, -1, 0, -1)\\
        (1, -1, 1, 1)  \end{tabular}
   &
     \begin{tabular}[t]{c}      (-1, 0, -1, 0)\\
        (-1, 1, 1, 0)  \end{tabular}
   &
    \begin{tabular}[t]{c}       (-1, -1, 1, 0)\\
        (0, 0, -1, 0)\\
        (0, 1, 1, 0)  \end{tabular}
   \\
    \hline
   24 	 & 25      & 26	     & 27 	  & 28 		 &29 \\
  \hline
   \begin{tabular}[t]{c}        (0, -1, 1, 0)\\
        (1, 0, -1, 0)\\
        (1, 1, 1, 0)  \end{tabular}
   &
        (1, -1, 1, 0)
   &
  \begin{tabular}[t]{c}         (-1, 0, -1, -1)\\
        (-1, 0, 0, 1)\\
        (-1, 1, 1, -1)  \end{tabular}
   &
  \begin{tabular}[t]{c}         (-1, -1, 1, -1)\\
        (0, 0, -1, -1)\\
        (0, 0, 0, 1)\\
        (0, 1, 1, -1)  \end{tabular}
   &
    \begin{tabular}[t]{c}       (0, -1, 1, -1)\\
        (1, 0, -1, -1)\\
        (1, 0, 0, 1)\\
        (1, 1, 1, -1)  \end{tabular}
   &
        (1, -1, 1, -1)
        \\
   \hline
         30 	 &        &  	     &   	  &   		 &   \\
  \hline
        (-1, 0, 0, 0) & & & &&\\
 \hline
\end{tabular}
}
  \caption{Elements of $\mathbb{Z}/31\mathbb{Z}$ are represented  as polynomials in $\gamma$, noted as vectors with lowest degree first. The reduction polynomial is $E(X)=X^4-2$ and 
 ${\gamma = 15}$ is a root of $E(X)$.  The digit set is $\{ -1, 0, 1\}$ (i.e., $\rho = 2$ ) }\label{ex:pmns1}
\end{table}
\end{example}

 We note that some values have more than one representation. This redundancy is not studied here, but it is useful in some applications \cite{DDEMMV2019}. Since all the elements of  $\mathbb{Z}/p\mathbb{Z}$ are represented, the value of $\rho$ satisfies $ \sqrt[n]{p}  \leq 2\rho - 1$, and redundancy starts when  $ \frac{\sqrt[n]{p}+1}{2} < \rho$.

\begin{remark}
  In~\cite{NP2008}, the authors proved that for every quadruple $(p,n,\gamma,\rho)$, there always exists a polynomial $E(X) \in \mathbb{Z}[X]$ satisfying $E(\gamma) \equiv 0 \mod p$, $\deg E(X) = n$ and $E(X)=X^n-c$ with $|c|\leqslant 2^{\frac{n}{2}}$.
However, one cannot hope to obtain fast primitives for modular arithmetic using a polynomial $E(X)$ with such a coefficient $c$ exponential in $n$. 
 Indeed, it is important to understand that modular operations are replaced in a PMNS by polynomial operations modulo $E(X)$, so that the degree of the result be still less than or equal to $n$. 
 The small size of the coefficients and the low density of the reduction polynomial $E(X)$ play a key role in the efficiency of modular reductions and in maintaining concise arithmetic.
 
Moreover, from a cryptographic point of view in the context of Side Channel Resistance, it could be of interest to build a PMNS from a polynomial $E(X)$ which has numerous roots modulo $p$,
since  distinct roots yield distinct  associated PMNS. In other words, from one execution to another one, for a fixed polynomial $E(X)$, a same secret value $k$ could be represented by a polynomial $K(X)$ which depends on the root used to build the PMNS.
  
Consequently, once the parameters $p$ and $n$ are given, or in other words, once it has been held that the integers modulo $p$ will be encoded on $n$ symbols,  the key question that arises is then which polynomials $E(X)$

 \begin{enumerate}
  \item allow one to find a parameter $\rho$ as small as possible,
   \item offer a good modular reduction,
  \item have a large number of roots $\gamma$ in $\mathbb{Z}/p\mathbb{Z}$.
 \end{enumerate} 
 
\noindent Next sections of this paper are devoted to these questions.  
\end{remark}

\section{Construction and specifications of PMNS}
\label{sec:theorem}
%\label{BoundPMNS}

In this section, we give conditions to ensure the existence of a PMNS $\GB=(p,n,\gamma,\rho)_E$ for a generic $E(X)$.

\begin{theorem}\label{theorem_PMNS_generalise}
  Let $p \geqslant 2$ and $n  \geqslant 2$ be two  integers, $E(X)$ be a monic polynomial of degree $n$ in $\mathbb{Z}[X]$ and $\gamma$ an integer which is a root of $E(X)$ in $\mathbb{Z}/p\mathbb{Z}$.\\
  Let $\GL$ be the $n$-dimensional lattice  generated by the polynomials in $\mathbb{Z}[X]$ of degree at most $n-1$ for which $\gamma$ is a root modulo $p$. This lattice $\GL$  is generated by the following $n\times n$ matrix $\mathbf{A}$ (with respect to the canonical monomial basis, with polynomials represented in lines)
  \begin{equation}
\label{matrixA}
\mathbf{A} = \begin{pmatrix}
p &0&\dots&\dots&0&0\\
 -\gamma &1 & \dots & \dots & 0 &0\\
\vdots &\ddots & \ddots &   &   & \vdots \\
0 & \dots & -\gamma & 1 & \dots &0 \\
\vdots &  &   & \ddots & \ddots & \vdots  \\
0 & 0 & \dots & \dots & - \gamma & 1
 \end{pmatrix}.
%\left(
%\begin{array}{cccccc}
%p &0&\dots&\dots&0&0\\
% -\gamma &1 & \dots & \dots & 0 &0\\
%\vdots &\ddots & \ddots &   &   & \vdots \\
%0 & \dots & -\gamma & 1 & \dots &0 \\
%\vdots &  &   & \ddots & \ddots & \vdots  \\
%0 & 0 & \dots & \dots & - \gamma & 1
%\end{array}\right).
\end{equation} 
 Then, \\
\indent for any $\rho>\mu_\infty(\GL)$ (the covering radius for the max-norm), \\
\indent the system $\GB = (p,n,\gamma,\rho)_E$ is a Polynomial Modular Number System.
 \end{theorem}
 
\begin{proof}
  Let $a \in [0,p-1]$ and let $T_a(X)=a$. 
We know that for any vector $T \in \mathbb{R}^n$ there exists $V \in \GL$ such that  $ \| T - V  \|_\infty \leqslant \mu_\infty(\GL)$. % \| T - V  \|_2 \leqslant r$.
Hence, there exists $V_a \in \GL$ such that $\|T_a - V_a \|_\infty \leqslant \mu_\infty(\GL) < \rho$, and $(T_a - V_a)(\gamma) \equiv T_a(\gamma)  - V_a(\gamma)  \equiv a \bmod{p}$ (since $V_a\in\GL$). In consequence, for any $a\in[0,p-1]$, $T_a - V_a$ is a polynomial which fulfills the condition of  \cref{def:pmns}.
We conclude that $\GB=(p,n,\gamma,\rho)_E$ is a PMNS.

\end{proof}

Currently, there is no efficient algorithm to compute the covering radius of a lattice. In the next two sections, we describe how to obtain an effective calculation of the bound on $\rho$. 

\subsection{Relation between the lattice's basis and the PMNS}

\begin{theorem} \label{RhoBoundTheorem}
  Let $p \geqslant 2$ and $n  \geqslant 2$ be two integers, $E(X)$ be a monic polynomial of degree $n$ in $\mathbb{Z}[X]$ and $\gamma$ be a root of $E(X)$ in $\mathbb{Z}/p\mathbb{Z}$.\\
  Let $\GL$ be the lattice  of polynomials in $\mathbb{Z}[X]$ of degree at most $n-1$, for which $\gamma$ is a root modulo $p$ 
 , $B$ a basis of $\GL$ and $\mathbf{B}$ the matrix associated to this basis (each row is an element of $B$).\\
Then, 
\[
 \mbox{for any }  \rho > \frac{1}{2}   \left\|  \mathbf{B^T} \right\|_\infty=  \max_{j}\left\{ \sum_{i=0}^{n-1} |b_{i,j} | \right\},
 \]
 \[ 
 \GB = (p,n,\gamma,\rho)_E \mbox{ is a Polynomial Modular Number System.}
 \]

\end{theorem}

\begin{proof}
Following the proof of \cref{theorem_PMNS_generalise}, we only have to show that for any polynomial $S(X)$, one can find a polynomial $T(X)\in\GL$ such that $\|S-T\|_{\infty}\leqslant\frac{1}{2}\|\mathbf{B}^T\|_\infty$.
Let $S \in  \mathbb{R}^n$. We define:
\begin{itemize}
\item $\lfloor  S \rceil $ as the vector whose coordinates are integers equal to the rounding to nearest integer of those  of  $S$;
\item $ \mbox{frac} (S)$ as the vector $ (S) = S - \lfloor  S \rceil $; notice that $\|  \mbox{frac} (S)\|_\infty \leqslant \frac{1}{2}\,.$
\end{itemize}

Let $S \in \mathbb{R}^n$.   We search a close vector $T \in \GL$ using a Baba\"i round-off approach \cite{Bab1986}.
We have,
$T =   \mathbf{B}^T \cdot \lfloor  ( \mathbf{B}^T)^{-1} \cdot S\rceil  $, thus

\[ 
S =    \mathbf{B}^T \cdot    (\mathbf{B}^T)^{-1} \cdot S =  T +   \mathbf{B}^T \cdot   \mbox{frac} \left(   (\mathbf{B}^T)^{-1} \cdot S \right)  \mbox{  with  } \left\|    \mbox{frac} \left(    (\mathbf{B}^T)^{-1} \cdot S  \right)\right\|_\infty \leqslant \frac{1}{2}.
\]

Then 

\[
 \left\|  S - T \right\|_\infty=  \left\|    \mathbf{B}^T\cdot   \mbox{frac} \left(   (\mathbf{B}^T)^{-1}\cdot S \right)\right\|_\infty \leqslant \frac{1}{2}   \left\|  \mathbf{B}^T \right\|_\infty.
 \]
\end{proof}

  In order to minimize $\rho$, a natural strategy is to choose a basis $B$ so that
$\|\mathbf{B^T}\|_\infty$ is small. 
  Such a basis can be computed from $A$ (\cref{theorem_PMNS_generalise}, eq.~\cref{matrixA}) using algorithms like LLL, BKZ or HKZ.

The next strategies can be applied when the polynomial $E(X)$ is irreducible.

\subsection{The case of irreducible reduction polynomials} 
Notice that \cref{RhoBoundTheorem} states that for any vector $S\in\mathbb{R}^n$, one can compute a vector $T$ in a lattice $\GL$ such that $\|S-T\|_\infty$ be smaller than $\frac{1}{2}\|\mathbf{B}^T\|_\infty$, where $B$ is a basis of $\GL$ and $\mathbf{B}$ its matrix form. The result holds for any lattice $\GL$ and any basis $B$ of this lattice. As a consequence, it can be applied to any basis $B'$ of a sublattice $\GL'$ of the lattice $\GL$ linked to the PMNS. The strategies described in this section are based on this remark.

Let $E(X)= X^n+a_{n-1}X^{n-1}+\dots+a_1X+a_0$, and let $\mathbf{C}$ be the companion matrix of $E(X)$:

 \begin{equation} \label{Eq:companion}
 \mathbf{C} =
  \begin{pmatrix}
   0 & 1 & 0 & \dots & 0 & 0 \\
  0 & 0 & 1 & \dots & 0 & 0 \\
  \vdots & \vdots & \vdots & \ddots & \vdots & \vdots \\
 0 & 0 & 0& \dots & 1 & 0 \\
  0 & 0 & 0 & \dots & 0 & 1\\
  -a_0& -a_1 &-a_2 &\dots & -a_{n-2} & -a_{n-1}
    \end{pmatrix}.
%   \footnotesize
%  \left(\begin{array}{cccccc}
%  0 & 1 & 0 & \dots & 0 & 0 \\
%  0 & 0 & 1 & \dots & 0 & 0 \\
%  \vdots & \vdots & \vdots & \ddots & \vdots & \vdots \\
% 0 & 0 & 0& \dots & 1 & 0 \\
%  0 & 0 & 0 & \dots & 0 & 1\\
%  -a_0& -a_1 &-a_2 &\dots & -a_{n-2} & -a_{n-1}
%  \end{array}\right).
  \end{equation}
Let $V = (v_0,\dots,v_{n-1})$ be the vector representing the coefficients of the polynomial $V(X) = \sum_{i=0}^{n-1} v_i X^i$, then $V.\mathbf{C}$ is the vector whose coordinates are the coefficients of the polynomial $X.V(X) \bmod E(X)$.

\begin{proposition}\label{PropInversible}
Let $V$ be a non-zero vector of $\GL$, the  lattice  of rank $n$ defined by $\mathbf{A}$  (\cref{theorem_PMNS_generalise}, eq. (\cref{matrixA})).
Let  $B_i = V\cdot \mathbf{C}^i$  be the row vector whose coordinates are the coefficients of the polynomial $B_i(X) = X^i\cdot V(X) \bmod E(X)$.
Let $\mathbf{B}$ be the $n\times n$ matrix  whose $i^{th}$ row is the vector  $B_i$.
\\

\noindent If $V(X)$ is inversible  modulo $E(X)$   then:
\begin{itemize}
\item the matrix $\mathbf{B}$ defines a sublattice  $\GL' \subseteq \GL$ of rank $n$ (i.e., $B=(B_0,\dots,B_{n-1})$ is a basis of  $\GL')$,
\item  and $V \in \GL'$.
\end{itemize}
\end{proposition}

\begin{proof}

The  $B_i$ are linearly independent. Indeed, let us suppose that there exists a non-zero vector $(t_0,t_1,\dots,t_{n-1}) \in \mathbb{Z}^n$ such that $\sum_{i=0}^{n-1} t_i B_i = 0 $. It means that $\sum_{i=0}^{n-1} t_i X^i V(X)= 0 \bmod{E(X)}$, or equivalently $T(X) V(X) = 0  \bmod{E(X)}$, with $T(X) = \sum_{i=0}^{n-1} t_i X^i $.
Then $T(X) V(X) V^{-1}(X)\bmod{E(X)} = T(X)=0 $, since $V(X)$ is inversible modulo $E(X)$ and  degree of $T(X)$ is at most $n-1$. Hence the rows of $\mathbf{B}$ are a basis of a sublattice  $\GL' \subseteq \GL$ of rank $n$, and  $V \in \GL'$ (which corresponds to the first row of B).
\end{proof}

\begin{corollary}\label{CorIrreduc}
  Let $V$ be a non-zero vector of $\GL$, the  lattice  of rank $n$ defined by $\mathbf{A}$  (\cref{theorem_PMNS_generalise}, eq. (\cref{matrixA})).

    If  $E(X)$ is irreducible, then 
\begin{itemize}
\item $V$   defines a sublattice  $\GL' \subseteq \GL$ of rank $n$, (i.e., $\mathrm{B}=(B_0,\dots,B_{n-1})$, defined in  \cref{PropInversible} is a basis of  $\mathfrak{L'} )$,
\item  moreover $V \in \GL'$.
\end{itemize}
\end{corollary}

\begin{proof}
If  $E(X)$ is irreducible, then $V(X)$ is inversible and  \cref{PropInversible} gives $B=(B_0,\dots,B_{n-1})$, a basis of  $\GL'$,   $\GL' \subseteq \GL$ of rank $n$, and  $V \in \GL'$.
\end{proof}

%%%%

A possible strategy to lower the bound on $\rho$ is then to take a short vector $V \in \GL$, that is, a vector which satisfies the Minkowski bound $  \| V\|_\infty \leqslant \alpha p^{1/n}$ with $\alpha \in \left] 0,1 \right]$. From this vector $V$, we build the basis $B$ of the sublattice $\GL'$ to compute the lower bound on $\rho$. 

In this context, we can provide a bound such that if $\rho$ is greater than this, then we are guaranteed to have a PMNS.
Let us consider the $(2n-1)\times n$ matrix $M$ whose rows are the coefficients of $X^i\ {\rm mod}\ {E(X)}$ for $0\leqslant i\leqslant 2n-2$. For any polynomial $T(X)$ of degree at most $2n-2$, the coefficients of $T(X)$ mod $E(X)$ are computed as the vector-matrix product $TM$.

\begin{proposition}\label{Prop:irreducbound}
Let $E(X)$ be an irreducible polynomial, let  $M$  be the $(2n-1)\times n$ matrix whose rows are the coefficients of $X^i\ {\rm mod}\ {E(X)}$, for $0\leqslant i\leqslant 2n-2$, and $s=\|M^T\|_\infty$. 
\[ 
\mbox{If  } \rho > \frac{1}{2}  p^{1/n}(1+(n-1)s) \; ( \geqslant  \frac{1}{2} \|\mathbf{B}^T\|_\infty ), 
\]
\[ 
\mbox{then } \GB = (p,n,\gamma,\rho)_E  \mbox{ is a Polynomial Modular Number System.}
\]
\end{proposition}
\begin{proof}

Let $V$ be a short vector of the lattice $\GL$, hence $\|V\|_\infty\leqslant p^{1/n}$. From  \cref{PropInversible}, the matrix $\mathbf{B}$ is a basis of a sublattice $\GL'$. Each row $B_i$ contains the coefficients of $X^iV(X)\ {\rm mod}\ E(X)$. These coefficients are computed as the vector-matrix product $T^{(i)}M$ where $T^{(i)}(X)=X^iV(X)$. Hence $\|B_i\|_\infty\leqslant   s\|V\|_\infty$ for $i\geqslant 1$, and 
$\|B_0\|_\infty\leqslant  p^{1/n}$.
Therefore $\|\mathbf{B}^T\|_\infty\leqslant p^{1/n}(1+(n-1)s)$. We conclude using  \cref{RhoBoundTheorem}.
\end{proof}

A second strategy is  to use the companion matrix $\mathbf{C}$ of $E(X)$ for  computing a basis $B$ of  $\GL'$ .

\begin{corollary}\label{CorIrreducb}
Let  $\GL$ be the  lattice  of rank $n$ given by $\mathbf{A}$  (\cref{theorem_PMNS_generalise}, eq. \cref{matrixA}), let $\mathbf{C}$ be the companion matrix of $E(X)$,
and let $\GL_D$ be  the lattice of rank $n$ in $\mathbb{Z}^{n^2}$ defined by
$\mathbf{ D}= (\mathbf{ A}  |\mathbf{ A \cdot C}^1 |\cdots |\mathbf{A \cdot C}^{n-1})$.

For any $\overline{V} = (V_0,V_1,\cdots,V_{n-1}) \in \GL_D$ 
such that  $\overline{V}  \neq (0)^{n^2}$,\\
\noindent
if  $E(X)$ is irreducible  then:\begin{enumerate}
\item  $V_0 \in  \GL$,
\item  $(V_0,V_1,\cdots,V_{n-1})$ is a basis of $\GL' \subseteq \GL$.

\end{enumerate}
\end{corollary}

\begin{proof}
$V_0$ is a linear combination of  rows of $\mathbf{A}$, hence it belongs to  $\GL$.
Next,  since $V_i = V_0 \cdot \mathbf{C}^i $, for all $i \geqslant 1$, then, due to  \cref{CorIrreduc},  the vector $(V_0,V_1,...,V_{n-1})$ is a basis of a sublattice   $\GL' \subseteq \GL$.
\end{proof}

Hence, the last strategy is to choose a short vector $(V_0,V_1,\cdots, V_{n-1})$ of  $\GL_D$ and to build the basis $B$ of  $\GL'$ from $V$.

\subsection{Some examples of PMNS}

In these examples we give the value of the norm $ \left\|  \mathbf{B^T}\right\|_\infty$ for each reduced basis approach: LLL \cite{LLL1982}or BKZ \cite{Sch1994} or HKZ reduction \cite{KZ1873,LLS1990} of  $\mathbf{A}$, or the one of    \cref{CorIrreduc}, or  \cref{CorIrreducb}.
We remark that the last two approaches offer  the best results for polynomials $E(X)$ with small coefficients.
In    \cref{NbSystems}, we give  experimental results  with exhaustive searches.

\begin{example}\mbox{}

{\footnotesize $p =  112848483075082590657416923680536930196574208889254960005437791530871071177777$}

$n =  8$, $E(X) =  X^8 + X^2 + X + 1$,

{\footnotesize $\gamma =  14916364465236885841418726559687117741451144740538386254842986662265545588774$}

\begin{tabular}{llll}
 LLL: & $ \left\|  \mathbf{B^T}\right\|_\infty = 16940155314$ &
 BKZ: & $ \left\|  \mathbf{B^T}\right\|_\infty = 15289909984$\\

                                       HKZ: & $ \left\|  \mathbf{B^T}\right\|_\infty = 15289909984$& &\\

                Cor.  \ref{CorIrreduc}: & $ \left\|  \mathbf{B^T}\right\|_\infty = 13881325101$&

 Cor.  \ref{CorIrreducb}: & $ \left\|  \mathbf{B^T}\right\|_\infty = \textcolor{black}{12883199915}$

\end{tabular}
\end{example}

\begin{example}\mbox{}

{\footnotesize $p =  96777329138546418411606037850670691916278980249035796845487391462163262877831$}

$n =  8$, $E(X) =  X^8 - X^4 - 1$,

{\footnotesize $\gamma =  66378119609141043317728290217053385256449145407556727004132373270146455575461$}

\begin{tabular}{llll}
LLL: & $ \left\|  \mathbf{B^T}\right\|_\infty = 17955608045$ &
 BKZ: & $ \left\|  \mathbf{B^T}\right\|_\infty = 17955608045$\\

                                      HKZ: & $ \left\|  \mathbf{B^T}\right\|_\infty = 17955608045$& &\\

                    Cor.  \ref{CorIrreduc}: & $ \left\|  \mathbf{B^T}\right\|_\infty = 11628752571$&

Cor.  \ref{CorIrreducb}: & $ \left\|  \mathbf{B^T}\right\|_\infty = \textcolor{black}{10489321362}$

\end{tabular}
\end{example}

\begin{example}\mbox{}

{\footnotesize $p =  94234089378179148303661339351342500658910595299680545500602453424882978290351$}

$n =  8$, $E(X) =  X^8 + X^4 - X^3 + 1$,

{\footnotesize $\gamma =  55857489577292751855009098551500852039618350925837275620376166398325678525151$}

\begin{tabular}{llll}
LLL: & $ \left\|  \mathbf{B^T}\right\|_\infty =  \textcolor{black}{12305954812}$ &
 BKZ: & $ \left\|  \mathbf{B^T}\right\|_\infty = 12305954812$\\

                     HKZ: & $ \left\|  \mathbf{B^T}\right\|_\infty = 12305954812$& &\\

    Cor.  \ref{CorIrreduc}: & $ \left\|  \mathbf{B^T}\right\|_\infty ={15570303402}$&

 Cor.  \ref{CorIrreducb}:  & $ \left\|  \mathbf{B^T}\right\|_\infty = 14857375293$

\end{tabular}
\end{example}

\begin{example}\mbox{}

{\footnotesize $p =  96777329138546418411606037850670691916278980249035796845487391462163262877831$}

$n =  8$, $E(X) =  X^8 + 6$,

{\footnotesize $\gamma =  5538274654329514802181726618906590237936295237553666062542808070676484572674$}

\begin{tabular}{llll}
LLL: & $ \left\|  \mathbf{B^T}\right\|_\infty = \textcolor{black}{12509178620}$ &
 BKZ: & $ \left\|  \mathbf{B^T}\right\|_\infty = 12509178620$\\

                     HKZ: & $ \left\|  \mathbf{B^T}\right\|_\infty = 12509178620$& &\\

   Cor.  \ref{CorIrreduc}: & $ \left\|  \mathbf{B^T}\right\|_\infty = {47611052126}$&

Cor.  \ref{CorIrreducb}:  & $ \left\|  \mathbf{B^T}\right\|_\infty = 40733847267$

\end{tabular}
\end{example}

\section{Suitable irreducible polynomials for PMNS}
\label{sec:classes}
\label{reduction_polynomials}

In  \cref{theorem_PMNS_generalise}, we proved that if $E(X)$ is an irreducible polynomial, then we can define a PMNS $\GB = (p,n,\gamma,\rho)_E$ depending of $E(X)$.
For efficiency reason on reducing modulo $E(X)$, $E(X)$ must respect some criteria, in particular with respect to the size of the digits in $\GB = (p,n,\gamma,\rho)_E$.
We thus define what can be a suitable PMNS irreducible reduction polynomial.

\subsection{Suitable PMNS reduction polynomial}
\begin{definition}\label{def:suitable}
A polynomial $E(X)$ is a suitable PMNS reduction polynomial, if:
\begin{enumerate}
\item  $E(X)$ is irreducible in $\mathbb{Z}[X]$,
\item  $E(X) = X^n + a_k X^k + \dots + a_1 X + a_0 \in \mathbb{Z}[X]$, with $n \geqslant 2$ and $k \leqslant \frac{n}{2}$,
\item most of the coefficients $a_i$ are zero,  other ones are very small (if possible equal to $\pm1$) compare to $p^{1/n}$.
\end{enumerate}
\end{definition}

The second item ensures that the polynomial reduction modulo $E(X)$ of a polynomial $T(X)$ of degree lower than $2n$ is done in two steps, i.e., by two times, we split $T(X) = T_1(X) X^n + T_0(X)$ with $T_1(X)$ and $T_0(X)$ of degree lower than $n$, and we substitute  $X^n \bmod E(X) = - (\sum_{i=0}^{k}a_iX^i) \bmod E(X)$.

The third item allows one to give a bound on the coefficients of  $T(X) \bmod E(X)$, namely $\| T(X) \bmod E(X) \|_\infty < s \| T(X)\|_\infty$, where $s$ is the $l1-$norm of the $(2n-1)\times n$ matrix $S$ whose row $i$ represents the coefficients of $X^{i} \pmod{E(X)}$ for $i=0\dots2n-1$ (see Prop. 2.3 of \cite{DDEMMV2019}). As a consequence, if $G(X)$ and $F(X)$  are  two elements of the PMNS,  i.e., $\|F(X)\|_\infty < \rho$ and $\|G(X) \|_\infty< \rho$, then $\|F(X) \times G(X) \|_\infty< n\rho^2$ and  $\|F(X) \times G(X)  \pmod{E(X)}\|_\infty< sn\rho^2$.

%\color{red}
\section*{Why consider alternatives for $E(X)$} Since the definition of the PMNS representation system, all the research focused on the polynomial $E(X)=X^n-\lambda$ because the external reduction can be efficiently performed when $\lambda$ is ``small'' (often a power of 2 to use logical operator) \cite{BIP2015s,NP2008,amns_09,amns_12,DDEMMV2019,DDV2020,sike,MPHELL,Negre_2021}. Now, from proposition \ref{Prop:irreducbound}, we know that the size of the coefficients used in the PMNS representation system depends on the parameter $s$ which in turn depends on the coefficients of the polynomial $E(X)$ since $s=\|M^T\|_\infty$ where $M$ is the $(2n-1)\times n$ matrix whose rows are the coefficients of $X^i\bmod E(X)$. Hence the smaller $s$ is, the smaller $\rho$ is. As a toy example, let us consider $n=6$ and the irreducible polynomial $E(X)=X^6+4$, then it is easy to see that $s=5$ since each column of $M$ contains only two elements (1 and -4), except the last one which contains only one element equal to 1. Now let us consider $E(X)=X^6-X-1$, then it is irreducible (see proposition \ref{prop:trinome}) and a simple computation gives $s=3$. This value for $s$ can also be obtained considering the polynomial $E(X)=X^6-2$ which corresponds to the AMNS case. In fact, for the AMNS case, one can see that $s=|\lambda|+1$, hence $s$ is proportional to $\lambda$. So, the only way to minimize $s$ is to take $\lambda=\pm 2$ (a simple argument shows that $\lambda=\pm 1$ does not allow to build an AMNS). Notice that the reduction modulo $X^6-X-1$ is very efficient and competitive with the one computed with $X^6-2$. Our goal to study suitable PMNS reduction polynomial is thus to enlarge the set of polynomials which can be used to define a PMNS without being restricted to the exclusive choise of the AMNS subset taking $\lambda=\pm 2$. We propose to developpers a set of polynomials for which the value $s$ can easily be computed so that depending on the context (software or hardware), they can select the better choice which fits their constraints.
\\
Another point of view concerns countermeasure to side channel attack. In the spirit of what has been proposed in \cite{BaImLi04a}, 
 one may consider to build for a fixed prime $p$ numerous PMNS representations. Let us consider the ECC context. Once $kP$ must be computed, first we choose the PMNS system to use, than we compute $kP$. This approach complements other countermeasures described in \cite{DDEMMV2019,Negre_2021}.
 Now, from a practical point of view, if we focus on the polynomials $E(X)=X^n-\lambda$ with $\lambda$ a power of 2, this will drastically reduce the choice of possible PMNS. Hence our goal is to enlarge the possible choice of PMNS for a prime $p$ by considering other polynomials $E(X)$ with small coefficients so that the external reduction can be efficiently performed and so that $\rho$ be small.
\\
\color{black}
According to the first item of \cref{def:suitable}, a suitable polynomial is irreducible. In the sequel, we adapt some classical irreducibility criteria and give examples of irreducible polynomials with few non-zero coefficients satisfying the two other items.

\subsection{Classical  polynomial irreducibility criteria}

To verify the first item of \cref{def:suitable}, we can use general criteria  such as the Sch\"onemann-Eisenstein criterion, Dumas'  criterion \cite{Dum1906} or the generalization given by N. C.   Bonciocat in \cite{Bon2015}. We adapt these criteria to our purpose, namely to a monic polynomial $E(X)= X^n + a_k X^k + \dots + a_1 X + a_0$, with  $k \leqslant \frac{n}{2}$.

\begin{proposition}[from Dumas' criterion \cite{Dum1906}] If there exists a prime $\mu$ and an integer $\alpha$  such that, $\mu^\alpha \mid a_0$, $\mu^{\alpha+1} \nmid a_0$,  $\mu^{\lceil \alpha (n-i)/n \rceil} \mid a_i$,
and $\gcd(\alpha,n) = 1$, then $E(X)= X^n + a_k X^k + \dots + a_1 X + a_0$ is irreducible over $\mathbb{Z}[X]$.
\end{proposition}

For example, $E(X)= X^n + \mu X^k +\mu$ is irreducible according to this criterion. If $k< n/2$ and $\mu << p^{1/n}$, then $E(X)$ is a suitable PMNS reduction polynomial.

\begin{proposition}[from Corollary 1.2  \cite{Bon2015}]
Let $E(X) = X^n + a_k X^k + \dots + a_1 X + a_0$, $a_0 \neq 0$, let $t \geqslant 2$ and
let $\mu_1,\dots,\mu_t$ be pairwise distinct numbers, and $\alpha_1,\dots,\alpha_t$ positive integers. 
If, for  $j = 1,\dots,t$,  and  $i = 0,\dots,k$,   $\mu_j^{\alpha_j} \mid a_i $, $\mu_j^{\alpha_j+1} \nmid a_0$, and  $\gcd(\alpha_1,\dots,\alpha_t,n) =1$, then $E(X)$ is irreducible over  $\mathbb{Z}[X]$.
\end{proposition}

For example, $E(X)= X^n + \mu_1^{\alpha_1}\mu_2^{\alpha_2} X^k +\mu_1^{\alpha_1}\mu_2^{\alpha_2}$, with $\gcd(\alpha_1,\alpha_2,n) =1$, is irreducible with this criterion.  If $k< n/2$ and $\mu_1^{\alpha_1}\mu_2^{\alpha_2} << p^{1/n}$, then $E(X)$ is a suitable PMNS reduction polynomial.

\subsection{Suitable   Cyclotomic Polynomials for PMNS}\label{cyclo_irreducible}
A well-known set of irreducible polynomials in $\mathbb{Z}\left[X\right]$ is the set of cyclotomic polynomials.
Let us denote by  ${\tt ClassCyclo(n)}$ the class of suitable cyclotomic polynomials for PMNS, whose degree is $n$.

\begin{proposition}\label{CycloProp} For $n>1$,  $ \varPhi_m(X) $ the $m$-th cyclotomic polynomial is a suitable polynomial  if and only if $\varphi(m)=n=2^i3^j$ with $i \geqslant 1,j \geqslant 0$.

 (i.e., ${\tt ClassCyclo(n)} \neq  \emptyset$ if and only if $n=2^i3^j$ with $i \geqslant 1,j \geqslant 0$.)
\end{proposition}

\begin{proof}
%{\color{blue}NEW PROOF

For $m >1$ $ \varPhi_m(X) $ the $m$-th cyclotomic polynomial,  is  self-reciprocal, $\varPhi_m(X) = X^{n}\varPhi_m(\frac{1}{X})$ with $n= \varphi(m) $ the degree of $\varPhi_m(X)$, (i.e.,the coefficients  $a_i $ of the term $X^i$ are equal to those $ a_{n-i}$  of the terms $X^{n-i}$ for all $i$).
Thus, suitable cyclotomic polynomials will be of the form  $X^{2n'} + a X^{n'} +1$ with $n=2n'$.

If $X_0$ is a root of a cyclotomic  $X^{2n'} + a X^{n'} +1$,   then $X^{n'}_0$ is a root of unity and  a root of $X^2+aX+1$,  as its conjugate too, hence we have $a=2\cdot cos \theta$. Since $a$ is an integer, we have $a=\pm2,\pm1,0$. But, for $a=\pm 2$ the polynomial  $X^{2n'} \pm 2 X^{n'} +1 = (X^{n'} \pm1)^2$ is not irreducible.
Therefore $a = \pm 1 , 0$.
\begin{itemize}
\item a = 0, we consider $  X^{n}+1$.  If $n=2^i\cdot t$ with $t>1$ odd then 
$X^n+1 = (X^{2^i} + 1) (X^{2^i\cdot (t-1)} + X^{2^i\cdot (t-2)}+\dots+1)$.
Thus $n=2^i$ and the cyclotomic polynomials are  $\varPhi_{2^{i+1}} (X)= X^{2^i} +1$.

\item $a=1$, then we look for  $ \varPhi_m  (X)  = X^{2n'}+ X^{n'} + 1$, and $e^{\frac{2i\pi}{3n'}}$ is a root of this polynomial.
$e^{\frac{2i\pi}{3n'}}$ is also a root of $X^{3n'}-1$.  We know that $X^{3n'}-1$ is the product of the cyclotomic $\varPhi_d(X)$ with $d | 3n'$ and $e^{\frac{2i\pi}{3}}$ is one of its roots, thus
$e^{\frac{2i\pi}{3}}$  is a root of $\frac{X^{3n'}-1}{X^{2n'}+ X^{n'} + 1}= X^{n'} - 1$. Hence $n'$ is a multiple of $3$ and by induction $n'=3^j$ and 
$ \varPhi_m  (X)  = \varPhi_{3^{j+1}}  (X)  = X^{2 \cdot 3^j}+ X^{3^j} + 1$ with $j \geq 0$.

\item $a=-1$, then we look for  $ \varPhi_m  (X)  = X^{2n'} - X^{n'} + 1$. 
Let $n'=2^{i-1}\cdot \alpha$, with $\alpha$ odd and $i\geq 1$, then 
$X^{2^{i}\alpha} - X^{2^{i-1}\alpha} + 1 = (-X^{2^{i-1}})^{2\alpha}+ (-X^{2^{i-1}})^\alpha +1$, we can refer to the previous case to deduce that $\alpha = 3^j$. Thus $ \varPhi_m  (X)  = \varPhi_{2^{i}3^{j+1}}  (X)  = X^{2^{i} \cdot 3^j}- X^{2^{i-1}3^j} + 1$ with $i\geq1$ and $j\geq 0$.
\end{itemize}

We have proved that for $n>1$,
if ${\tt ClassCyclo(n)} \neq  \emptyset$ then $n=2^i3^j$ with $i \geqslant 1,j \geqslant 0$.
\\

Reciprocally, for $n=2^i3^j$, $i\geqslant 1$, $j\geqslant 0$, we have to show that there exists a suitable cyclotomic polynomial whose degree is $n$. 
\\

Let $n=2^i$ ($i\geqslant 1$), since $n=\varphi (m)$, then $m=2^{i+1}$ and $\varPhi_{2^{i+1}}(X) = X^{2^i}+1$ is a suitable cyclotomic polynomial.

Let $n=2.3^j$ ($j\geqslant 1$), since $n=\varphi (m)$, then $m=3^{j+1}$ and $\varPhi_{3^{j+1}}(X) = X^{2.3^j}+X^{3^j}+1$ is a suitable cyclotomic polynomial.

Let $n=2^i3^j$ ($i\geqslant 2$, $j\geqslant 1$), since $n=\varphi (m)$, then $m=2^i3^{j+1}$ and $\varPhi_{2^i3^{j+1}}(X) = X^{2^i3^j}-X^{2^{i-1}3^j}+1$ is a suitable cyclotomic polynomial.
%The reciprocal is  obtained using the fact that  $ \varPhi_m  (X)$ with $m=\prod_{i=1}^{k} p_i^{\alpha_i}$ is such that  $ \varPhi_m  (X)  =  \varPhi_{m_0}  (X^{\frac{m}{m_0}})$ with $m_0=\prod_{\alpha_i\neq 0} p_i$, thus $m_0=2,3,6$ and $m=2^i, 2\cdot 3^j, 2^i\cdot 3^{j+1}$ give the suitable cyclotomic polynomials.
\end{proof}

\subsection{Suitable reduction  $\{-1,1\}$-quadrinomials}
In \cite{FJ2006}, Finch and Jones give criteria of irreducibility for polynomials $X^a+\beta X^b+\gamma X^c+\delta$ with $\beta,\gamma, \delta \in\{ -1, 1\}$ and $a>b>c>0$.

\begin{proposition}[Theorem 2 in \cite{FJ2006} ]
The quadrinomial  $X^a+\beta X^b+\gamma X^c+\delta$  with $\beta,\gamma, \delta \in\{ -1, 1\}$ and $a>b>c>0$, is irreducible over $\mathbb{Z}[X]$\\
 if and only if $\gcd(a,b,c) = 2^t m$ with $m$ odd, 
 and it satisfies one of the following conditions  :
\begin{enumerate}
  \item $(\beta,\gamma,\delta) = (1,1,1)$ and $\overline{a}\overline{b}\overline{c}\equiv 1 \pmod{2}$,
    \item $(\beta,\gamma,\delta) = (-1,1,1)$, $b'-c' \not\equiv 0 \pmod{2 \overline{a}}$, $b' \not\equiv 0 \pmod{2 \overline{b}}$ and $a'-b' \not\equiv 0 \pmod{2 \overline{c}}$, 
    \item $(\beta,\gamma,\delta) = (1,-1,1)$, $b'-c' \not\equiv 0 \pmod{2 \overline{a}}$, $a'-c' \not\equiv 0 \pmod{2 \overline{b}}$ and $c' \not\equiv 0 \pmod{2 \overline{c}}$,
  \item $(\beta,\gamma,\delta) = (1,1,-1)$,  $a' \not\equiv 0 \pmod{2\overline{a}}$,  $b' \not\equiv 0 \pmod{2\overline{b}}$ and  $c' \not\equiv 0 \pmod{2\overline{c}}$,
  \item $(\beta,\gamma,\delta) = (-1,-1,-1)$, $a' \not\equiv 0 \pmod{2 \overline{a}}$, $a'-c' \not\equiv 0 \pmod{2 \overline{b}}$ and $a'-b' \not\equiv 0 \pmod{2 \overline{c}}$.
  \end{enumerate}
  Where $a' = a/2^t$,  $b' = b/2^t$, $c' = c/2^t$ and  $\overline{a}=\gcd(a',b'-c')$,  $\overline{b}=\gcd(b',a'-c')$,  $\overline{c}=\gcd(c',a'-b')$.
  
  We call this  class of suitable reduction quadrinomials ${\tt ClassQuadrinomials}$, and ${\tt ClassQuadrinomials(n)}$ is the set of such quadrinomials of degree $n$.

\end{proposition}

For example, $E(X) =X^{2^t7m}+X^{2^t3m}+X^{2^tm}+1$, with $m$ odd,  is a suitable PMNS reduction quadrinomial verifying the first condition.

\subsection{Suitable  reduction  $\{-1,1\}$trinomials}

In this part we refer to   a paper  of W.H. Mills \cite{Mil1985b} and one of W. Ljunggren \cite{Lju1960}. The first one gives a criterion on quadrinomials and roots of unity, the second one gives an application to trinomials.

\begin{proposition}  
\label{prop:trinome}
We note $\gcd(n,m)=d$ and $n= d \cdot n_1$, $m=d \cdot m_1$. If $n_1+m_1 \not\equiv 0 \mod 3$, then the polynomial $X^n+ \beta X^m + \delta$ with $\delta, \beta \in \{  -1,1\}$ and $n > 2m > 0$ is irreducible over $ \mathbb{Z}[X]$.

The class of the suitable reduction trinomials verifying these criteria is named ${\tt ClassTrinomials}$,  and ${\tt ClassTrinomials(n)}$ represents the set of the trinomials of degree $n$.

\end{proposition}

\begin{proof}
Let us  transform, like in \cite{Lju1960}, $E(X) = X^n+ \beta X^m + \delta$  in quadrinomial:
\[
\left(X^n+ \beta X^m + \delta\right)\left(X^n - \delta\right) = X^{2n}+ \beta X^{n+m} - \beta\delta X^m -1 = F(X).
\]

Theorem 2  of  \cite{Mil1985b} states that if  $F(X) = A(X)E(X)$, where every root of $A(X)$ and no root of $E(X)$ is  a root of unity, then $E(X)$ is irreducible except if there exists $r$ such that:
\begin{itemize}
\item $(2n, n+m, m)  = (8r,7r,r)$ and $(\beta,\delta)= (1,-1)$ or $(-1,-1)$,
\item or $(2n, n+m, m)   = (8r,4r,2r)$  and $(\beta,\delta)= (1,-1)$,
\item or $(2n, n+m, m)   = (8r,6r,4r)$ and $(\beta,\delta)= (-1,-1)$.
\end{itemize}

It is easy to check that there is no integer $r$ which satisfies any of these 3 constraints, hence we only have to verify that no root of $E(X)$ is a root of unity. First notice that, because $n=d n_1$ and $m=d m_1$ with $\gcd(n_1,m_1) =1$, if $\lambda$ is a root of $E(X)$, then $\lambda^d$ is root of $X^{n_1}+ \beta X^{m_1}+\delta$. Hence, if the roots of $X^{n_1}+ \beta X^{m_1}+\delta$ are not roots of unity, then no root of $E(X)=X^n+ \beta X^m + \delta$ is a root of unity.

Let us assume that $\lambda$ is a root of $X^{n_1}+ \beta X^{m_1}+\delta$, which is also a root of unity.
Then there exit $t>1$ and $k$ with $\gcd(k,t)=1$, such that:
\[ 
\lambda = e^{\frac{2ik\pi}{t}} = \cos{\frac{2k\pi}{t}} + i \sin{\frac{2k\pi}{t}}. 
\]

 Assume that $\beta =1$. Then
   \[
   \left\{\begin{array}{ l}
      \cos(\frac{2n_1k\pi}{t}) + \cos(\frac{2m_1k\pi}{t})=  2 \cos(\frac{k\pi(n_1+m_1)}{t}) \cos(\frac{k\pi(n_1-m_1)}{t})      = - \delta    \\
      \sin(\frac{2n_1k\pi}{t}) + \sin(\frac{2m_1k\pi}{t})  =   2 \sin(\frac{k\pi(n_1+m_1)}{t}) \cos(\frac{k\pi(n_1-m_1)}{t})    =0.
\end{array}\right.
\]

Last equality implies that  $\sin(\frac{k\pi(n_1+m_1)}{t})  =0$ or $\cos(\frac{k\pi(n_1-m_1)}{t}) =0$.
Since $\delta \neq 0$,  the first equation implies that $\cos(\frac{k\pi(n_1-m_1)}{t}) \neq 0$, hence $\frac{k(n_1+m_1)}{t}$ is an integer. Since  $\gcd(k,t) =1$, $t \mid (n_1+m_1)$.
 This last result implies that the first equation can be reduced to 
 \[
  \cos\left(\frac{k\pi(n_1-m_1)}{t}\right) =\pm \frac{1}{2}
  \]
  because $\delta = \pm 1$.
 
 It means that 
\[
 \frac{k\pi(n_1-m_1)}{t}= j \frac{\pi}{3},  \; \;\; j=1,2,4,5  \pmod{6}.
 \]
Hence, $ t \mid 3 (n_1 - m_1$), since $\gcd{(k,t)} =1$.
  
Assume that $\beta =-1$. The system becomes:
   \[
   \left\{\begin{array}{ l}
      \cos(\frac{2n_1k\pi}{t}) - \cos(\frac{2m_1k\pi}{t})=  - 2 \sin(\frac{k\pi(n_1+m_1)}{t}) \sin(\frac{k\pi(n_1-m_1)}{t})      = - \delta    \\
      \sin(\frac{2n_1k\pi}{t}) - \sin(\frac{2m_1k\pi}{t})  =   2 \cos(\frac{k\pi(n_1+m_1)}{t}) \sin(\frac{k\pi(n_1-m_1)}{t})    =0.  
\end{array}\right.
\]

The first equation implies that $\sin(\frac{k\pi(n_1-m_1)}{t})  \neq 0$, hence the second equation gives   $\frac{k\pi(n_1+m_1)}{t} = j \frac{\pi}{2}$ for $j$ odd, which implies $t \mid 2(n_1+m_1)$.
Since $\frac{k\pi(n_1+m_1)}{t} = j \frac{\pi}{2}$ for $j$ odd, then the first equation can be reduced to $\sin(\frac{k\pi(n_1-m_1)}{t})  = \pm \frac{1}{2}$, which means that 

\[
\frac{k\pi(n_1-m_1)}{t} =j \frac{\pi}{6}, \;\;\; j = 1,5,7,11 \pmod{12}.
\]
Hence $ t \mid 6 (n_1 - m_1)$.\\

To sum up, if $\lambda$ is a $t^{th}$ root of unity of $X^{n_1}+ \beta X^{m_1}+\delta$ with $\delta,\beta \in \{-1,1  \}$, then:
\begin{itemize}
  \item[(a)] if $\beta = 1$, $t \mid (n_1 + m_1)$ and $ t \mid 3(n_1 - m_1)$,
  \item[(b)] if $\beta = -1$, $t \mid 2(n_1 + m_1)$ and $ t \mid 6(n_1 - m_1)$.
\end{itemize}

The case (a) implies that if $3 \nmid t$, then  $t \mid (n_1 - m_1)$, thus $t\mid 2n_1$ and $t \mid 2m_1$, as $\gcd(n_1,m_1) = 1$. We conclude that  $t=2$ and $\lambda =1$ or $-1$ is a root of $E(X)$ which is impossible.

The case (b) implies that  if $3 \nmid t$, then $t \mid 2 (n_1 - m_1)$, thus $t=4$, and  $\lambda =i$, $-i$, $1$ or $-1$ is a root of $E(X)$, which is impossible.

Hence, if one root of $E(X)$ is a root of unity, then $3$ divides $t$, and  $n_1+m_1 \equiv 0 \mod 3$.

In conclusion, if $\gcd(n_1,m_1)=1$  and $n_1+m_1 \not\equiv 0 \mod 3$, then $X^{n_1}+ \beta X^{m_1} + \delta$ and  $X^{n}+ \beta X^{m} + \delta$ are irreducible.
\end{proof}

\subsection{Case of irreducibility of binomials $X^n + c$, $c \in \mathbb{Z}$, $|c| \geqslant 2$, over $\mathbb{Z}$}

\begin{proposition}\label{section_irr_binome}
Let $|c| =\prod_{j=1}^{k} p_j^{m_j}$ with $p_j$ pairwise distinct prime numbers, and $m_j$ positive integers.
If $gcd(m_1,\dots,m_k,n) = 1$,  then the polynomial $X^n + c$, with $c \in \mathbb{Z}$, $|c| \geqslant 2$,  is irreducible over $\mathbb{Z}[X]$.

We call this class of suitable polynomials ${\tt ClassBinomial}$, and, for $n$ and $c$ satisfying this proposition, ${\tt ClassBinomial(n,c)}$ is the singleton $\{X^n + c\}$.
\end{proposition}

\begin{proof}
It is a direct application of  Corollary 1.2 of a paper due to 
Nicolae Ciprian Bonciocat \cite{Bon2015}.
\end{proof}

%%%%%%%%%%%%%%%%%%%%%%%%
\subsection{Polynomials with bounds on the modules of their complex roots}

The two propositions given in this section are inspired by the Perron irreducibility criterium, which is proved thanks to Rouch\'e's theorem \cite{Bon2010}.

\begin{proposition}
For a fixed $n \geqslant 2$ and  a prime $\mu$,  let  $P(X)= X^n + \displaystyle\sum_{i = 1}^{n/2} \varepsilon_i X^i \pm \mu $ with $\varepsilon_i \in \{ -1, 0, 1\}$.

If $\displaystyle \mu > 1 + \sum_{i=1}^{n/2} \left|\varepsilon_i \right|  $,  then the polynomial $P(X)$  is irreducible over $\mathbb{Z}[X]$.\\

They represent the fifth class of  suitable  reduction polynomials. We call this class ${\tt ClassPrimeCst}$, and  ${\tt ClassPrimeCst(n,\mu)}$ represents all the polynomials of this class with $n \geqslant 2$ and $\mu$ a prime number.
\end{proposition}

\begin{proof}
Since  $\displaystyle \mu > 1 + \sum_{i=1}^{n/2} \left|\varepsilon_i \right|  $,    there exists $\delta >1$ such that  $\displaystyle \mu > \delta^n \left(1 +  \sum_{i=1}^{n/2} \left|\varepsilon_i \right| \right)$.

Let us consider  $\mathcal{C} = \{z \in \mathbb{C}\;  / \;  |z| = \delta  \}$, $P(X)= X^n + \displaystyle\sum_{i = 1}^{n/2} \varepsilon_i X^i + \varepsilon \mu $  ($\varepsilon_i \in \{ -1, 0, 1\},\; \varepsilon \in \{ -1, 1\} $),  $F(X) = \varepsilon  \mu$ and $G(X) = P(X)-F(X)$.

For any $z \in \mathcal{C}$, we have $\displaystyle |G(z)| \leqslant  \delta^n \left(1 + \sum_{i=1}^{n/2} \left|\varepsilon_i \right|  \right) < \mu = |F(z)|$.

Since $F(z)$ and $G(z)$ are holomorphic functions, Rouch\'e's theorem states that $F(z)$ and $P(z)=F(z)+G(z)$ have the same number of roots inside $\mathcal{C}$. Hence $P(z)$ has no root inside $\mathcal{C}$ since $F(z)$  is constant. In other words, any root $\alpha$ of $P(z)$ satisfies $|\alpha|\geqslant \delta > 1$.

Assume now that $P(X)$ is reducible over $\mathbb{Z}\left[ X\right]$. Hence,  $P(X) = H(X) Q(X)$ with $H(X)$ and $Q(X)$ two monic polynomials.
Since  $|P(0) |= \mu$ (a prime number), we can assume that $|H(0) |= \mu$ and $|Q(0) |= 1$.
Now  $\prod \left|z_i\right| =1$, where $z_i$ are all the roots of  $Q(X)$.  But the roots of $Q(X)$ are also roots of $P(X)$ which is  not possible since any root $\alpha$ of $P(X)$ is such that $|\alpha | \geqslant \delta > 1$.
Hence, $P(X)$ is irreducible over  $\mathbb{Z}\left[ X\right]$.
\end{proof}

\begin{remark} If  $\mu > n/2+1$, then ${\tt ClassPrimeCst(n,\mu)}$ contains $3^{n/2}$ elements (for each $\varepsilon_i $ three possibilities), else
$\displaystyle\sum_{i=0}^{\mu-2} \left( \begin{array}{c} n/2 \\ i \end{array} \right) 2^{i+1}$ elements.
\end{remark}

\begin{proposition}
For a fixed $n \geqslant 2$,   let  $P(X)= X^n + \displaystyle\sum_{i = 2}^{n/2} \varepsilon_i X^i  + a_1 X  \pm 1 $ with $\varepsilon_i \in \{ -1, 0, 1\}$ and $a_1 \in\mathbb{Z}^*$.

If $\displaystyle \left| a_1 \right|  > 2 + \sum_{i=2}^{n/2} \left| \varepsilon_i \right|$,  then the polynomial $P(X)$  is irreducible over $\mathbb{Z}[X]$.\\

We call this class ${\tt ClassPerron}$, and  ${\tt ClassPerron(n,a_1)}$ represents all the polynomials of this class with $n \geqslant 2$, $a_1 \in \mathbb{Z^*}$.
\end{proposition}

\begin{proof}
The proof is similar to the previous one.
From $|a_1| > 2 + \sum_{i=2}^{n/2} |\varepsilon_i|$, we can deduce that there exists $\delta > 1$ such that $|a_1| > \delta^n \left(    2 + \sum_{i=2}^{n/2} |\varepsilon_i| \right)$. 
Then, from  Rouch\'e's theorem, $P(z)$ and $F(z) = a_1 z$ have the same number of roots inside $\mathcal{C} = \{z \in \mathbb{C}\;  / \;  |z| = \delta  \}$.
Hence $P(z)$ has only one root whose module is strictly less than $\delta$.

Now, if $P(X)$ is reducible over  $\mathbb{Z}\left[ X\right]$, then  $P(X) = H(X) Q(X)$, with $H(X)$ and $Q(X)$ two monic polynomials and  $\left| H(0) \right| = \left| G(0) \right| = 1$.
Hence, $H(z)$ has at least one root $z_H$ such that $|z_H| \leqslant 1$ and  $G(z)$ has at least one root $z_G$ such that $|z_G| \leqslant 1$. It means that $P(z)$ has at least two roots inside $\mathcal{C}$, which is not possible. 
Hence, $P(X)$ is irreducible over  $\mathbb{Z}\left[ X\right]$.
\end{proof}

\begin{remark} If  $|a_1|> n/2 + 1$, then ${\tt ClassPerron(n,a_1)}$ contains $2 \times 3^{n/2-1}$ elements, else
$\displaystyle\sum_{i=0}^{|a_1|-3} \left( \begin{array}{c} n/2 -1\\ i \end{array} \right) 2^{i+1}$ elements.
\end{remark}

%%%%%%%%%

\section{Number of PMNS in function of their reduction polynomial  in $\mathbb{Z}/p\mathbb{Z}$ with $p$ prime\label{Sec_Roots}}
\label{sec:number}

In this section, we determine for each class,  the reduction polynomials which have one or more roots $\gamma$ in $\mathbb{Z}/p\mathbb{Z}$. 
The number of  roots in $\mathbb{Z}/p\mathbb{Z}$ defines the number of  possible PMNS. 

As we have to build, for a given prime $p$ and a given  number of digits $n$,  many  PMNS with an efficient arithmetic,  finding relevant reduction polynomials is crucial. 
Now that we have described classes of irreducible polynomials with specific reduction properties, we need to identify for a prime $p$ which polynomials have at least one root in $\mathbb{Z}/p\mathbb{Z}$, and if possible, how many. We begin with a presentation of two special cases where the reduction polynomials are cyclotomics or binomials, then we propose a method in the general case that works for any irreducible integer polynomial.

\subsection{Number of PMNS with a cyclotomic reduction polynomial}

\begin{proposition}\label{prop_roots_cyclo}
Let $p$ be a prime number, $p >2$, and an integer $m\geqslant 3$ such that $m \mid (p-1)$. Then the cyclotomic polynomial  $\Phi_m(X) $ satisfies $\Phi_m(X) \mid (X^{p-1}-1)$ and $\Phi_m(X) $ has $\varphi(m)$  roots over $\mathbb{Z}/p\mathbb{Z}$. 
\end{proposition}

\begin{proof} We have
$\displaystyle (X^{p-1}-1) = \prod_{\xi_i \in (\mathbb{Z}/p\mathbb{Z})^*}^{} (X - \xi_i ) = \prod_{d | (p-1)}^{} \Phi_d(X)$.

Thus $\displaystyle \Phi_m(X) \mid (X^{p-1}-1)$, and  $\Phi_m(X) $ has $\varphi(m)$  (its degree) roots over $\mathbb{Z}/p\mathbb{Z}$.
\end{proof}

We apply  \cref{prop_roots_cyclo} to the different cyclotomic polynomials of  the class ${\tt ClassCyclo(n)}$ introduced in  \cref{CycloProp}.

\begin{corollary}
Let $p$ be a prime number, $n \geqslant 2$ such that $n = 2^i 3^j$, with $i, j \in \mathbb{N}$.  

If either one of these conditions folds, i.e.;
\begin{itemize}
\item[$a)$]  $i> 0$, $j = 0$,   $(2\,n)$ divides $(p-1)$, and $E(X) = \Phi_{2n}(X) = X^{n} +1$;
\item[$b)$]  $i = 1$, $j \geqslant 0$,   $(3\,n\,/\,2)$ divides $(p-1)$, and $E(X) = \Phi_{\frac{3n}{2}}(X) = X^{n} + X^{\frac{n}{2}} + 1$;
\item[$c)$]  $i \geqslant 1$, $j \geqslant 0$,   $(3\,n)$ divides $(p-1)$, and $E(X) = \Phi_{3n}(X) = X^{n} - X^{\frac{n}{2}} + 1$,
\end{itemize}
then,  there exist $n$ PMNS $(p,n,\gamma_i,\rho)_{E(X)}$, with $\gamma_i$ one of the $n$ distinct roots modulo $p$ of ${E(X)}$.

\end{corollary}

\begin{example}  
 Construction  of PMNS from a cyclotomic reduction polynomial for $p = 2^{256}\cdot 3^{157}\cdot 115+1$ coded on $512$ bits.
 \begin{itemize} 
 \item $E(X) = X^8 + 1$: from its eight roots, the best $\rho$   is obtained with  \cref{CorIrreduc}  and   \cref{CorIrreducb}., and it is   $66$ bits number. 
 \item $E(X) = X^6+X^3 + 1$: from its six roots, the best $\rho$   is obtained twice with  LLL, else with  \cref{CorIrreduc}  and   \cref{CorIrreducb}, and it is $87$ bits number.
 \item $E(X) = X^6 -X^3 + 1$: from  its six roots, the best $\rho$  is obtained with  \cref{CorIrreduc}  and   \cref{CorIrreducb}, and it is $87$ bits number.
  \end{itemize} 
 
\end{example}

\subsection{Number of PMNS with reduction binomials $X^n + c$, $c \in \mathbb{Z}$, $|c| \geqslant 2$}

\begin{proposition}\label{proposition_PMNS_binomials}
Let $E(X) = X^n + c$  be an element of  ${\tt ClassBinomial(n,c)}$ (\cref{section_irr_binome}).
Let  $g$ be a generator of $(\mathbb{Z}/p\mathbb{Z})^{\times}$ and $y$ such that $g^y \equiv -c \mod p$.

If $\gcd(n, p-1)$ divides $ y$,
then $E(X) = X^n + c$ has  $\gcd(n, p-1)$ different roots.
\end{proposition}

\begin{proof}
Let $X_0$ be a solution of $E(X) =0 \pmod{p}$. Then there exists $x_0$ such that $X_0 \equiv g^{x_0} \pmod{p}$ 
and $ g^{n\cdot x_0} \equiv -c  \equiv g^y \pmod{p}$.
In other words, $n \cdot x_0 \equiv y \pmod{p-1}$.

Now, let $\delta = \gcd(n,p-1)$. A classical result in modular arithmetic states that this linear equation admits $\delta$ solutions if and only if $\delta$ divides $y$, each solution being equal to $x_0+jp'$, where $j\in \left\{0, \dots,\delta -1  \right\}$ and $(p-1) = \delta p'$.
\end{proof}

\begin{remark} If $\gcd(n, p-1)= 1$, then $E(X) = X^n + c$ is guaranted to have  one  root.
\end{remark}

%%%

\begin{example}
For  $p=40993$,  $5$ is a generator of $\left( \mathbb{Z}/40993\mathbb{Z} \right)^*$.
Let $n=4$ and $E(X) = X^4 +c$. For $c=2$, we can find $y= 33788$ such that  $-c =5^y \mod p$. Since $\gcd(1,n)=1$, from  \cref{proposition_PMNS_binomials}, $E(X)$ is irreducible. Moreover, $\gcd(n,p-1)=4$ divides $y$, hence four PMNS can be generated from $E(X)$.
For $c'=-2$, we can find $y'=13292$ and $\gcd(n,p-1)=4$ divides $y'$, giving once again four possible PMNS.
\end{example}

\subsection{Number of PMNS in the general case}

In this part, we propose a general method to count the minimum number of PMNS we can reach from a prime $p$ and any irreducible polynomial $E(X)$ in $\mathbb{Z}[X]$.

\begin{proposition}\label{number_of_PMNS}
Let $p$ be a prime number, $n > 2$, $E(X)$ a polynomial of degree $n$ and irreducible in $\mathbb{Z}[X]$, and $D(X) = \gcd(X^p - X, E(X)) \mod p$.

There exist $\deg(D(X))$ Polynomial Modular Number Systems $(p, n, \gamma_i, \rho)_{E(X)}$.

\end{proposition}

\begin{proof}
The proof is immediate considering, when $p$ is prime, that the roots of $X^p - X \mod p$ are the $p$ elements of  $\mathbb{Z}/p\mathbb{Z}$. 
\end{proof}

\begin{remark}
 \Cref{prop_roots_cyclo} can be considered as a corollary of   \cref{number_of_PMNS}.
\end{remark}

The computation of $ \gcd(X^p - X, E(X)) \mod p = \gcd(X^{p-1} - 1, E(X)) \mod p$ ($E(X)$ is irreducible in $\mathbb{Z}/p\mathbb{Z}$) can be done, in a reasonable time, in two steps:\begin{enumerate}
  \item  we compute $X^{p-1}\bmod E(X) \bmod p$ with a square and multiply exponentiation algorithm, and we compute $F(X) = X^{p-1} - 1  \bmod E(X) \bmod p,$
  \item  then, we compute $D(X)= \gcd(F(X), E(X)) \bmod p$ with polynomials of degrees lower than or equal to $n.$
\end{enumerate}
The first step represents   $O(\log_2(p))$ squares and additions of polynomials of degree lower than $n$ in  $\mathbb{Z}/p\mathbb{Z}[X]$, and the second step represents at most $n$ iterations of the Euclidean algorithm. \\

The roots can be found using the method of  Cantor-Zassenhaus\cite{CZ1981} for separating the roots of $D(X)=\gcd(X^p - X, E(X)) \bmod p$. 

As $X^p - X = \prod_{z\in \mathbb{Z}/p\mathbb{Z}} (X - z)$, then  $D(X)=\prod_{i=1}^{k} (X-e_i)$ with $k=\deg(D(X))$ and $e_i \in\mathbb{Z}/p\mathbb{Z}$ all distinct.

Due to the Chinese Remainder Theorem, any polynomial $A(X)$ of degree strictly lower than $k$, can be represented by its values modulo the $(X-e_i)$: 
\[
a_i = A(X) \bmod{(X-e_i)} \mbox{ in } \mathbb{Z}/p\mathbb{Z} \mbox{ for } i = 1,\dots,k\,.
\]
Let us consider a polynomial $A(x)$ such that $a_i\in\{0,1,-1\}$  and  $A(X) \neq 0,1,-1  \bmod D(X)$ (i.e., $a_i$ are not all equal).
We note $T=\{i, a_i = 1\}$ and $S=\{i, a_i = 0\}$. As $ A(X) \neq 0,1,-1  \bmod D(X)$, at least one of this two sets is not empty with a cardinal strictly lower than $k$. We can obtain a proper factor of $D(X)$ by computing;
\[
\gcd(D(X), A(X)-1) = \prod_{i\in T} (X - e_i)\  \mbox{or}\ \gcd(D(X), A(X)) = \prod_{i\in S} (X - e_i)\,.
\]

To find such a polynomial $A(X)$, we consider a random polynomial $B(X)\in  (\mathbb{Z}/p\mathbb{Z})[X]$ of degree lower than $k$. We note $b_i = B(X) \bmod (X-e_i)$ in  $\mathbb{Z}/p\mathbb{Z}$. Then, 
\[b_i^{p-1} = 0,1\  \mbox{and}\  b_i^{\frac{p-1}{2}} = 0,1,-1 \mbox{ in } \mathbb{Z}/p\mathbb{Z}\,. 
\]
If  $B(X)^{\frac{p-1}{2}}  \neq 0,1,-1  \bmod D(X)$, then we choose $A(X) =  B(X)^{\frac{p-1}{2}}   \bmod D(X)$. 
If $\gcd(D(X), A(X)-1) $ and  $\gcd(D(X), A(X))$ are trivial factors, then we draw randomly another polynomial $B(X)$, else we iterate this method with the found non trivial factors $\gcd(D(X), A(X)-1) $, $\gcd(D(X), A(X))$ and $D(X)$ divided by these factors, until all the factors are of degree $1$.

\begin{example} We consider $p = 7826474692469460039387400099999297$ and the reduction polynomial  $E(X)  =  X^5 + X^2 + 1$.
First, we compute
\[R(X) \begin{array}[t]{l}= X^{p-1} - 1 \bmod E(X) \mbox{ in }(\mathbb{Z}/p\mathbb{Z})[X]\\
      = 3659189086300930014207106583318421 \; X^4    \\
      + 7322126259420098177093985099094624\; X^3 \\
      +    1727826215301243349042222461135262\; X^2\\
      + 7098030983909056985211630090182831\; X\\
      + 7372958503626664659096728485020294
\end{array}
\]
Then we obtain
 \[D(X) \begin{array}[t]{l} 
= \gcd(R(X), E(X))  \mbox{ in }(\mathbb{Z}/p\mathbb{Z})[X]\\
 = X^2  
 + 1305849998419067291000337897705258 \; X \\
 + 1793073000954204546034194068098826
\end{array}  
\]
Next, we randomly draw $B(X)$  $\bmod\ D(X)$  in $(\mathbb{Z}/p\mathbb{Z})[X]$,\\
$$\begin{array}[t]{ll} 
B(X)= &7090634213741414696606254289781859  \; X \\
&+ 4896184070237294585014544822120651 
\end{array}  \\
$$
We compute 
$
A(X) \begin{array}[t]{l}
=B(X)^{\frac{p-1}{2}} \bmod D(X)  \mbox{ in }(\mathbb{Z}/p\mathbb{Z})[X]\\
=6630612051164461204925113188582895  \; X \\
+ 7099602401400966247478428555087365
\end{array}\\
$
We obtain a first factor $\mbox{ in }(\mathbb{Z}/p\mathbb{Z})[X]$,
$$T(X)= gcd(A(X)-1,D(X)) = X+ 2974625651330718059716669102633643.$$
By division we find the second factor,\\
$$D(X) / T(X) = X - 1668775652911650768716331204928385$$
%We draw a second  $B(X)$ randomly  $\bmod D(X)$,\\
% $B(X)=  \begin{array}[t]{l}
% 3680931181210936929914943699271633 \; X\\
% + 4164507889698283718083397281811470
%\end{array}
% $\\
%We compute 
%$
%A(X) \begin{array}[t]{l} 
%= B(X)^{\frac{p-1}{2}} \bmod D(X) \\
%=1195862641304998834462286911416402 \; X \\
%+ 726872291068493791908971544911932
%\end{array}
%$\\
%Then $gcd(A(X)-1, D(X)) = x + 6157699039557809270671068895070912$\\
%

%We obtain two roots of         
%$\begin{array}[t]{l} 
%E(X) \bmod p:   \\  
%\gamma_1=1668775652911650768716331204928385\\
%\gamma_2 =  4851849041138741979670730997365654.
%\end{array}  $

\end{example}

%%%%%%%%%%%%%%%%%%%%%%%%%%%%%%%
\subsection{Example giving all the possible PMNS for a given $p$\label{NbSystems}}

This example was produced with SageMath subroutines for the $256$-bits prime $p$:\\
 {\footnotesize $p =57896044618658097711785492504343953926634992332820282019728792003956566811073$}, 
 and $n=9$. We consider the  PMNS  $\GB = (p,n,\gamma,\rho)_E$  such that:
 \begin{itemize}
 \item  $E(X) = X^9 + a_k X^k + \dots + a_1 X + a_0 \in \mathbb{Z}[X]$, where $k \leqslant 4$,
 \item   the coefficients $a_i$  satisfy  $| a_i | \leqslant 1$ for $1 \leqslant i\leqslant k$ and  $| a_0 | \leqslant 3$,
  \item $ \rho \leqslant 2^{31}$.  
\end{itemize}

The number of PMNS  $\GB = (p,n,\gamma,\rho)_E$  that can be built for different polynomials verifying the  criteria is equal to $354$.

 Most of the time, the best $\rho$ is obtained   266 times by LLL  but    BKZ or HKZ are  46 times better than LLL , then 42 are better than the previous  ones with  \cref{CorIrreduc}  or    \cref{CorIrreducb}  or  \cref{PropInversible}  with a short vector.

\section{Conclusion}
\label{sec:conclusion}
In this paper, we have shown with    \cref{theorem_PMNS_generalise} the link between the existence of a PMNS and the lattice generated by its reduction polynomial and its modulo. We thus set a bound on the size of the PMNS digits depending on the covering radius of this lattice. 
Then,  \cref{RhoBoundTheorem} provides  a bound which can   easily be computed from the infinity norm of a  basis of the lattice. This second theorem has led us to consider PMNS defined by an irreducible polynomial. In this case,  it is easy to define a basis of the lattice that can be associated with the PMNS (\cref{PropInversible},   \cref{CorIrreduc} and \cref{CorIrreducb}). These results allowed us to produce PMNS with specific reduction polynomials allowing efficient reductions and whose roots give the bases ($\gamma$) of these systems. Now, we have the opportunity to offer for a given modulo $p$ a wide variety of PMNS with small digits and reduced associated lattices.

Very recently, the use of PMNS to perform modular multiplications was reintroduced in \cite{HvdH2019},  where some interesting complexity theoretical bounds are given.

\subsection*{Acknowledgment}  We thank Mrs. Val\'erie Berth\'e  for her attentive proofreading and her judicious suggestions.
\subsection*{Funding}This work was partially supported by the  ANR project ARRAND 15-CE39-0002-01 and the INRIA international associated team MACAO. All examples were coded with SageMath \url{https://www.sagemath.org}

\bibliographystyle{amsplain}
\bibliography{biblio5}

\providecommand{\bysame}{\leavevmode\hbox to3em{\hrulefill}\thinspace}
\providecommand{\MR}{\relax\ifhmode\unskip\space\fi MR }
\providecommand{\MRhref}[2]{%
  \href{http://www.ams.org/mathscinet-getitem?mr=#1}{#2}
}
\providecommand{\href}[2]{#2}
\begin{thebibliography}{10}

\bibitem{Ajtai98}
M.~Ajtai, \emph{The shortest vector problem in {\it $l_{\mbox{2}}$} is
  {NP}-hard for randomized reductions (extended abstract)}, Thirtieth Annual
  ACM Symposium on the Theory of Computing (STOC 1998), 1998, pp.~10--19.

\bibitem{AAACDKLMMPPRS20}
G.~Alagic, J.~Alperin-Sheriff, D.~Apon, D.~Cooper, Q.~Dang, J.~Kelsey, Y.-K.
  Liu, C.~Miller, D.~Moody, R.~Peralta, R.~Perlner, A.~Robinson, and
  D.~Smith-Tone, \emph{Status report on the second round of the {NIST}
  post-quantum cryptography standardization process}, Tech. Report NIST IR
  8309, National Institute of Standards and Technology, July 2020.

\bibitem{Bab1986}
L.~Babai, \emph{On {Lov\'asz}' lattice reduction and the nearest lattice point
  problem}, Combinatorica \textbf{6} (1986), no.~1, 1--13.

\bibitem{BD2021}
J.C. Bajard and S.~Duquesne, \emph{Montgomery-friendly primes and applications
  to cryptography}, Journal of Cryptographic Engineering (2021).

\bibitem{BIP2015a}
J.C. Bajard, L.~Imbert, and T.~Plantard, \emph{Arithmetic operations in the
  polynomial modular number system}, 17th IEEE Symposium on Computer Arithmetic
  (ARITH'05), IEEE, 2005, pp.~206--213.

\bibitem{BIP2015s}
\bysame, \emph{Modular number systems: Beyond the {Mersenne} family}, Selected
  Areas in Cryptography, Springer, 2005, pp.~159--169.

\bibitem{BaImLi04a}
Jean-Claude Bajard, Laurent Imbert, Pierre-Yvan Liardet, and Yannick Teglia,
  \emph{Leak resistant arithmetic}, Cryptographic Hardware and Embedded Systems
  - CHES 2004 (Berlin, Heidelberg) (Marc Joye and Jean-Jacques Quisquater,
  eds.), Springer Berlin Heidelberg, 2004, pp.~62--75.

\bibitem{VEB1981}
P.~Van~Emde Boas, \emph{Another {NP}-complete problem and the complexity of
  computing short vectors in lattices}, Tech. Report 81-04, Mathematics
  Department, University of Amsterdam, 1981.

\bibitem{Bon2010}
N.~C. Bonciocat, \emph{On an irreducibility criterion of {Perron} for
  multivariate polynomials}, Bull. Math. Soc. Sci. Math. Roumanie
  \textbf{53(101)} (2010), no.~3, 213--217.

\bibitem{Bon2015}
\bysame, \emph{{Sch{\"o}nemann--Eisenstein--Dumas}-type irreducibility
  conditions that use arbitrarily many prime numbers}, Journal Communications
  in Algebra \textbf{43} (2015), no.~8.

\bibitem{BF2001}
D.~Boneh and M.~Franklin, \emph{Identity-based encryption from the {W}eil
  pairing}, CRYPTO 2001, LNCS, vol. 2139, Springer-Verlag, 2001, pp.~213--229.

\bibitem{BCHL2013}
J.~W. Bos, C.~Costello, H.~Hisil, and K.~E. Lauter, \emph{Fast cryptography in
  genus 2}, {EUROCRYPT} 2013 (Thomas Johansson and Phong~Q. Nguyen, eds.),
  LNCS, vol. 7881, Springer, 2013, pp.~194--210.

\bibitem{BDKLLSSSS18}
J.~W. Bos, L.~Ducas, E.~Kiltz, T.~Lepoint, V.~Lyubashevsky, J.M. Schanck,
  P.~Schwabe, G.~Seiler, and D.~Stehle, \emph{Crystals - kyber: A cca-secure
  module-lattice-based kem}, 3rd IEEE European Symposium on Security and
  Privacy,, 2018, pp.~353--367.

\bibitem{sike}
Cyril Bouvier and Laurent Imbert, \emph{An alternative approach for sidh
  arithmetic}, Public-Key Cryptography -- PKC 2021 (Cham) (Juan~A. Garay, ed.),
  Springer International Publishing, 2021, pp.~27--44.

\bibitem{CZ1981}
D.~G. Cantor and H.~Zassenhaus, \emph{A new algorithm for factoring polynomials
  over finite fields}, Mathematics of Computation \textbf{36} (1981), no.~154.

\bibitem{Cas1959}
J.~W.~S. Cassels, \emph{An introduction to the geometry of numbers}, Classics
  in Mathematics, Springer-Verlag, 1959.

\bibitem{Asma2021}
Asma Chaouch, Laurent{-}St{\'{e}}phane Didier, Fangan{-}Yssouf Dosso, Nadia~El
  Mrabet, Belgacem Bouallegue, and Bouraoui Ouni, \emph{Two hardware
  implementations for modular multiplication in the {AMNS:} sequential and
  semi-parallel}, J. Inf. Secur. Appl. \textbf{58} (2021), 102770.

\bibitem{MPHELL}
Titouan Coladon, Philippe Elbaz-Vincent, and Cyril Hugounenq, \emph{Mphell: A
  fast and robust library with unified and versatile arithmetics for elliptic
  curves cryptography}, 2021 IEEE 28th Symposium on Computer Arithmetic
  (ARITH), 2021, pp.~78--85.

\bibitem{CS1988}
J.~H. Conway and N.~J.~A. Sloane, \emph{Sphere packings, lattices and groups},
  Grundlehren der mathematischen Wissenschaften, Springer-Verlag, 1988 (Third
  edition 1999).

\bibitem{Cra1992}
R.~Crandall, \emph{Method and apparatus for public key exchange in a
  cryptographic system}, {U.S.} Patent number 5159632, 1992.

\bibitem{AKSV2018}
J.-P. D'Anvers, A.~Karmakar, S.~Sinha~Roy, and F.~Vercauteren, \emph{Saber:
  Module-lwr based key exchange, cpa-secure encryption and cca-secure kem},
  AFRICACRYPT 2018, vol. 10831 LNCS, 2018, pp.~282--305.

\bibitem{DDEMMV2019}
L.-S. Didier, F.-Y. Dosso, N.~El~Mrabet, J.~Marrez, and P.~V{\'e}ron,
  \emph{{Randomization of Arithmetic over Polynomial Modular Number System}},
  {26th IEEE International Symposium on Computer Arithmetic} (Kyoto, Japan),
  vol.~1, {IEEE Computer Society}, 2019, pp.~199--206.

\bibitem{DDV2020}
L.-S. Didier, F.-Y. Dosso, and P.~V{\'e}ron, \emph{Efficient modular operations
  using the adapted modular number system}, Journal of Cryptographic
  Engineering (2020).

\bibitem{Dum1906}
G.~Dumas, \emph{Sur quelques cas d'irreductibilit{\'e} des polyn{\^o}mes {\`a}
  coefficients rationnels}, Journal de Math{\'e}matique Pure et Appliqu{\'e}e
  \textbf{2} (1906).

\bibitem{amns_12}
Nadia {El Mrabet} and Nicolas Gama, \emph{Efficient multiplication over
  extension fields}, {WAIFI}, Lecture Notes in Computer Science, vol. 7369,
  Springer, 2012, pp.~136--151.

\bibitem{amns_09}
Nadia {El Mrabet} and Christophe N{\`{e}}gre, \emph{Finite field multiplication
  combining {AMNS} and {DFT} approach for pairing cryptography}, {ACISP},
  Lecture Notes in Computer Science, vol. 5594, Springer, 2009, pp.~422--436.

\bibitem{FJ2006}
C.~Finch and L.~Jones, \emph{On the irreducibility of
  $\{-1,0,1\}$-quadrinomials}, INTEGERS: Electronic Journal of Combinatorial
  Number Theory \textbf{6} (2006).

\bibitem{Gal2012}
S.D. Galbraith, \emph{Mathematics of public key cryptography}, Cambridge
  University Press.

\bibitem{GMR2004}
V.~Guruswami, D.~Micciancio, and O.~Regev, \emph{The complexity of the covering
  radius problem on lattices and codes}, IEEE Conference on Computational
  Complexity, 2004, pp.~161--173.

\bibitem{Ham2012}
M.~Hamburg, \emph{Fast and compact elliptic-curve cryptography}, {IACR}
  Cryptol. ePrint Arch. \textbf{2012} (2012), 309.

\bibitem{HPS2011}
G.~Hanrot, X.~Pujol, and D.~Stehl{\'e}, \emph{Algorithms for the shortest and
  closest lattice vector problems}, International Conference on Coding and
  Cryptology, Springer, 2011, pp.~159--190.

\bibitem{HS2007}
G.~Hanrot and D.~Stehle, \emph{Improved analysis of {K}annan's shortest lattice
  vector algorithm}, CRYPTO, 2007.

\bibitem{HvdH2019}
D.~Harvey and J.~van~der Hoeven, \emph{Faster integer multiplication using
  short lattice vectors}, Thirteenth Algorithmic Number Theory Symposium ANTS
  XIII (msp, ed.), 2019.

\bibitem{HPS1998}
J.~Hoffstein, J.~Pipher, and J.~H. Silverman, \emph{{NTRU}: A ring-based public
  key cryptosystem}, Algorithmic Number Theory (ANTS 1998), LNCS, vol. 1423,
  Springer, 1998, pp.~267--288.

\bibitem{SIKE2019}
D.~Jao, R.~Azarderakhsh, M.~Campagna, C.~Costello, L.~De Feo, B.~Hess,
  A.~Jalali, B.~Koziel, B.~LaMacchia, P.~Longa, M.~Naehrig, G.~Pereira,
  J.~Renes, V.~Soukharev, and D.~Urbanik, \emph{{SIKE}: Supersingular isogeny
  key encapsulation}, Submission to the NIST's post-quantum cryptography
  standardization process, 2019.

\bibitem{Kob1987}
N.~I. Koblitz, \emph{Elliptic curve cryptosystems}, Mathematics of Computation
  \textbf{48} (1987), no.~177, 243--264.

\bibitem{KZ1873}
A.~Korkine and G.~Zolotareff, \emph{Sur les formes quadratiques}, Mathematische
  Annalen \textbf{6} (1873), pages366--389.

\bibitem{LLS1990}
J.~C. Lagarias, H.~W. Lenstra, and C.~P. Schnorr, \emph{Korkin-zolotarev bases
  and successive minima of a lattice and its reciprocal lattice}, Combinatorica
  \textbf{10} (1990), no.~4, 333--348.

\bibitem{LLL1982}
A.~K. Lenstra, H.~W. Lenstra, and L.~{Lov\'asz}, \emph{Factoring polynomials
  with rational coefficients}, Mathematische Annalen, Springer-Verlag
  \textbf{261} (1982), 513--534.

\bibitem{Lju1960}
W.~Ljunggren, \emph{On the irreducibility of certain trinomials and
  quadrinomials}, Mathematica Scandinavica \textbf{volume 8} (1960), no.~$n^o$
  1, 65--70.

\bibitem{Lov1986}
L.~{Lov\'asz}, \emph{An algorithmic theory of numbers, graphs and convexity},
  {CBMS}-{NSF} Regional Conference Series in Applied Mathematics, vol.~50, SIAM
  Publications, 1986.

\bibitem{Mil1985}
V.~S. Miller, \emph{Use of elliptic curves in cryptography}, {CRYPTO'85}, LNCS,
  vol. 218, Springer-Verlag, 1985, pp.~417--426.

\bibitem{Mil1985b}
W.~H. Mills, \emph{The factorization of certain quadrinomials}, Mathematica
  Scandinavica \textbf{57} (1985).

\bibitem{Min1896}
H.~Minkowski, \emph{Geometrie der zahlen}, B. G. Teubner, Leipzig, 1896.

\bibitem{Mon1985}
P.~L. Montgomery, \emph{Modular multiplication without trial division},
  Mathematics of Computation \textbf{44} (1985), no.~170, 519--521.

\bibitem{nist}
Dustin Moody, Gorjan Alagic, Daniel Apon, David Cooper, Quynh Dang, John
  Kelsey, Yi-Kai Liu, Carl Miller, Rene Peralta, Ray Perlner, Angela Robinson,
  Daniel Smith-Tone, and Jacob Alperin-Sheriff, \emph{Status report on the
  second round of the nist post-quantum cryptography standardization process},
  2020-07-22 2020.

\bibitem{NIST2009}
{National Institute for Standards and Technology}, \emph{{Digital Signature
  Standard (DSS)}}, Jun 2009.

\bibitem{NP2008}
C.~Negre and T.~Plantard, \emph{Efficient modular arithmetic in adapted modular
  number system using lagrange representation}, Proc. ACISP 08, Springer, 2008.

\bibitem{Negre_2021}
Christophe Negre, \emph{Side channel counter-measures based on randomized
  {AMNS} modular multiplication}, Proceedings of the 18th International
  Conference on Security and Cryptography, {SCITEPRESS} - Science and
  Technology Publications, 2021.

\bibitem{Pla2005}
T.~Plantard, \emph{{Arithm{\'e}tique modulaire pour la cryptographie}}, Theses,
  {Universit{\'e} Montpellier II - Sciences et Techniques du Languedoc}, 2005.

\bibitem{PSZ2015}
T.~Plantard, W.~Susilo, and Z.~Zhang, \emph{{LLL} for ideal lattices:
  re-evaluation of the security of gentry--halevi's fhe scheme}, Designs, Codes
  and Cryptography \textbf{volume 76} (2015), no.~$n^o$ 2, 325--344.

\bibitem{PFHKLPRSWZ17}
T.~Prest, P.-A. Fouque, J.~Hoffstein, P.~Kirchner, V.~Lyubashevsky, T.~Pornin,
  T.~Ricosset, G.~Seiler, W.~Whyte, , and Z.~Zhang, \emph{Falcon: Fast-fourier
  lattice-based compact signatures over {NTRU}}, Submission to the NIST's
  post-quantum cryptography standardization process, 2017.

\bibitem{RSA1978}
Ronald~L. Rivest, Adi Shamir, and Leonard~M. Adleman, \emph{A method for
  obtaining digital signatures and public-key cryptosystems}, Communications of
  the {ACM} \textbf{21} (1978), no.~2, 120--126.

\bibitem{Sch1994}
C.-P. Schnorr, \emph{Block reduced lattice bases and successive minima},
  Combinatorics, Probability {\&} Computing \textbf{3} (1994), 507--522.

\bibitem{Sol1999}
J.~A. Solinas, \emph{Generalized {M}ersenne numbers}, Research Report
  CORR-99-39, Center for Applied Cryptographic Research, University of
  Waterloo, Waterloo, ON, Canada, 1999.

\bibitem{SP2018}
D.~R. Stinson and M.~Paterson, \emph{Cryptography theory and practice}, fourth
  edition ed., Chapman and Hall/CRC, 2018.

\end{thebibliography}

\end{document}